\documentclass{article}

\usepackage[preprint]{neurips_2025}

\makeatletter
\renewcommand{\@noticestring}{}
\makeatother

\usepackage[utf8]{inputenc} 
\usepackage[T1]{fontenc}    
\usepackage{hyperref}       
\hypersetup{hidelinks}

\usepackage{url}            
\usepackage{booktabs}       
\usepackage{amsfonts}       
\usepackage{nicefrac}       
\usepackage{microtype}      
\usepackage{xcolor}         

\usepackage{graphicx}
\usepackage{multirow}
\usepackage{amsmath}
\usepackage{amsthm}
\usepackage{mathrsfs}
\usepackage[title]{appendix}
\usepackage{textcomp}
\usepackage{manyfoot}
\usepackage{algorithm}
\usepackage{algorithmicx}
\usepackage{algpseudocode}
\usepackage{listings}
\usepackage{float}

\definecolor{codebg}{HTML}{F2F2F2}
\definecolor{coderule}{HTML}{C8C8C8}
\lstdefinestyle{paperpy}{
  language=Python,
  basicstyle=\ttfamily\footnotesize,
  keywordstyle=\bfseries,
  commentstyle=\itshape,
  stringstyle=\relax,
  identifierstyle=\relax,
  backgroundcolor=\color{codebg},
  frame=single,
  rulecolor=\color{coderule},
  framerule=0.5pt,
  framesep=4pt,
  xleftmargin=8pt,
  xrightmargin=8pt,
  showstringspaces=false,
  breaklines=true,
  breakatwhitespace=true,
  columns=fullflexible,
  keepspaces=true,
  tabsize=2,
  aboveskip=0.7em,
  belowskip=0.7em,
  captionpos=b,
  numbers=none
}
\newcommand{\R}{\mathbb{R}}

\renewcommand{\SS}{\mathcal{S}}
\newcommand{\eps}{\varepsilon}

\definecolor{foadgreen}{RGB}{30,150,0}

\newif\ifrevisions
\revisionsfalse
\ifrevisions
  
\else
  
\fi

\providecommand{\Description}[1]{}

\newtheorem{theorem}{Theorem}

\theoremstyle{definition}

\title{Sampling for Region-Aggregated Spatial Scan Statistics}

\author{%
  Foad Namjoo \\
  University of Utah \\
  \texttt{foad.namjoo@utah.edu} \\
  \And
  Drew McClelland \\
  Taptap Send \\
  \texttt{drewmac6191@gmail.com} \\
  \And
  Michael Matheny \\
  Meta \\
  \texttt{michaelmathen@gmail.com} \\
  \And
  Jeff M. Phillips \\
  University of Utah \\
  \texttt{jeffp@cs.utah.edu} \\
}

\makeatletter
\g@addto@macro\@floatboxreset\centering
\makeatother

\begin{document}

\maketitle

\begin{center}\small
Accepted at ACM SIGSPATIAL 2026.
\end{center}

\begin{abstract}

Anomaly detection in geospatial data is a crucial tool in geographic information science (GIS), with applications ranging from national security to public-health surveillance to the study of societal disparities.  This work focuses on spatial scan statistics and addresses a key mismatch:
spatial counts are typically aggregated into predefined regions (census tracts, zip codes, counties), whereas the most efficient scan algorithms operate on spatial point data.
The standard remedy---collapsing each region to its centroid, as in widely used tools such as SaTScan---is convenient but, as we show, discards the region's spatial extent and causes a significant loss in statistical power.
To resolve this, we propose a simple yet scalable fix: replace each spatial region with 20--50 points sampled uniformly from its geometry,
and divide the region's measured and baseline counts evenly among them.
This approach improves statistical power while maintaining computational tractability. 
A convergence analysis explains why so few samples per region suffice. We recommend this sampling-based conversion as the default way to apply point-based spatial scan statistics to region-aggregated data for anomaly detection.
\footnote{Thanks to NSF CCF-2115677, IIS-1816149, and IIS-2311954.}

\end{abstract}

\section{Introduction}\label{sec1}

Scan statistics~\citep{kulldorff1997spatial,xie2022statistically,abolhassani2021up,glaz2024handbook} 
are a fundamental tool in spatial data analysis. They aim to identify spatial regions where a measured value is significantly different from what would be expected based on a baseline distribution.  The statistic evaluates this by ``scanning" all candidate spatial regions to detect the most anomalous ones.  
There are many variants of spatial scan statistics~\citep{NM04,NMC06,patil2004upper,Kulldorff2006,han2019kernel}, and they have been successfully used in various applications such as detecting elevated breast cancer rates~\citep{Kulldorff2006}, monitoring emerging public health threats~\citep{nobles2022presyndromic,desjardins2020rapid}, and analyzing spatial crime patterns \citep{shiode2011street}.   These tools are central to anomaly detection workflows in GIS and public health surveillance.

Scan statistics can be applied to spatial data provided in two formats.  The first format represents data as a set of geolocated points $X \subset \mathbb{R}^2$; such as home addresses or GPS-tracked locations.  The most common software for these tasks operates on an input of this form, such as SaTScan~\citep{Kul7.0}.  
The second common format aggregates the data into predefined regions (e.g., census tracts, zip codes, counties)~\citep{xie2022statistically,tango2005flexibly}.  The baseline and measured values can then be accumulated for each region. These aggregations may be necessary due to privacy constraints (an individual's data can be less easily identified if it has been region aggregated) or because data are only available at a coarse spatial resolution. 

To apply point-based algorithms to region-aggregated data, it is common practice to convert each region into a single point using its centroid. Although computationally convenient, this modeling choice introduces approximation errors and can significantly reduce the \emph{statistical power}---the ability to reliably detect a true anomaly.  
We make this precise in Section~\ref{sec:exp}.
In particular, it does not capture the spatial diversity and extent of the underlying region, leading to weaker anomaly detection.

This paper addresses the above limitation and proposes a robust yet scalable alternative.  \emph{The proposed solution is \underline{simple}: \textbf{instead of representing a region by one centroid point, we sample multiple points (typically 20 to 50) in each region
and divide the region's measured and baseline counts uniformly across those points.}} 
We show that this approach significantly increases the statistical power of these methods under planted region experiments and demonstrate its scalability on datasets at local, state, and national levels.
\footnote{Code and rendering scripts: \url{https://github.com/foadnamjoo/sampling-region-scan}}
By bridging region-aggregated data with efficient point-based algorithms, our method improves anomaly detection accuracy and strengthens spatial decision-making reliability.

\section{Preliminaries}



This section reviews both the classical point-based formulation and the challenges that arise when extending it to region-based settings.

\paragraph*{Point-Based Spatial Scan Statistics}
In the point-based setting, spatial scan statistics operate on a set of spatial points $X \subset \R^2$, where each point $x \in X$ has an associated baseline value $b(x)$,  which represents the expected background level of the measured phenomenon at location $x$ (e.g., population or expected case count) and a measured value $m(x)$ representing the observed quantity of interest (e.g., people with cancer).  The goal is to identify regions where the measured values deviate from the baseline, using a family of geometric shapes $\SS$ and a cost function $\phi$.

The family $\SS$ is typically defined by a geometric shape.  Common choices are disks (interiors of circles) or (axis-aligned) rectangles. Each shape $S \in \SS$  induces a subset $Y = X \cap S$ of points, and the scan statistic evaluates each induced subset $Y$ using the anomaly score function $\phi(Y)$.

A widely used cost function  $\phi(Y)$, is derived under a Poisson model~\citep{kulldorff1997spatial} for the distribution of measured values $m(x)$ given a baseline value $b(x)$.  An anomalous subset $Y \subset X$ is one that would be well modeled by a different (often larger) rate of measured values than $X \setminus Y$.  The final anomaly score $\phi(Y)$ is derived as the log-likelihood ratio between the best model with a different rate for $Y$ and $X \setminus Y$ and for the best model where the rate is the same for all $X$:
\[
\phi(Y) = m(Y) \log \frac{m(Y)}{b(Y)} + (1 - m(Y)) \log \frac{1 - m(Y)}{1 - b(Y)},
\]

where, $m(Y) = \frac{1}{M} \sum_{x \in Y} m(x)$ and $b(Y) = \frac{1}{B} \sum_{x \in Y} b(x)$, where $M = \sum_{x \in X} m(x)$ and $B = \sum_{x \in X} b(x)$. Although other cost functions exist~\citep{Kul7.0,agarwal2006spatial} (e.g., under Bernoulli or Gaussian assumptions), we adopt the Poisson model due to its widespread use and compatibility with both point and region-based inputs.



The final step of spatial scan statistics is to identify the shape $S$ that (approximately) maximizes $\phi(X \cap S)$.  This defines the spatial region that is most anomalous, and the statistic itself is   
\[
\Phi_{\SS}(X,m,b) = \max_{S \in \SS} \phi(X \cap S).
\]
While there are infinite number of geometric shapes that we could scan over, we only need to consider a finite number of them.  
This is because the function $\phi$ only depends on which points are contained in the shape, so the number of shapes considered is bounded as a function of the number of input points $n = |X|$.  Although there is an exponential number of subsets $2^n$, restricting to geometric shapes reduces this to a polynomial number: at most $n^3$ for disks and at most $n^4$ for rectangles.  
Brute-force enumeration of these candidates, spending $O(n)$ to
score each one, therefore takes $O(n^4)$ and $O(n^5)$ time for
disks and rectangles, respectively; 
sophisticated algorithms reduce this to near-quadratic $O(n^2 \log n)$ for rectangle discrepancy~\citep{APV06}, which the \texttt{pyScan} implementation we use in Section~\ref{sec:exp} exploits
(see Section~\ref{sec:runtime} for measured runtimes).


\paragraph*{Region-Based Spatial Scan Statistics}
In many real-world datasets, data are aggregated over spatial regions (e.g., counties), rather than available at individual locations. Let $Z$ be a collection of regions, with each region $z \in Z$ associated with the baseline and measured values $b(z)$ and $m(z)$, respectively. 
For any subset $Y \subset Z$, we define $m(Y) = \frac{1}{M}\sum_{z \in Y} m(z)$ and $b(Y) = \frac{1}{B}\sum_{z \in Y} b(z)$, where $M = \sum_{z \in Z} m(z)$ and $B = \sum_{z \in Z} b(z)$; the same cost function $\phi(Y)$ then applies unchanged.


The key challenge is to define the candidate subsets $Y$ (regions of interest).  While a potential subset could be any connected set of regions, as in FlexScan~\citep{tango2005flexibly,FleXScan}, this is not tied to geometric shapes, and can potentially over-fit to strangely shaped regions.  Also, it can be significantly slower than other scanning algorithms~\citep{grubesic}; also see Table~\ref{tab:runtime_usa+AR} in Section \ref{sec:runtime}.  

A more common strategy is to use geometric shapes $\SS$ (e.g., disks, rectangles) as before, but apply them to regions. This presents two main obstacles:
First, in the point-based setting, one can identify a canonical shape from a valid subset $Y \subset \R^2$ as the \emph{smallest} shape $S_Y \in \SS$ that contains $Y$ (and nothing from $X \setminus Y$).  Given a set of points $Y$ and $\SS$ as disks or rectangles, the smallest shape can be found easily with computational geometry.  However, with region-based input, where each $z \in Z$ has a complex geometry (e.g., a polygon in a shapefile), the computation of the minimal enclosing shape is more difficult.
Second, geometric shapes $S \in \SS$ are likely to be unable to partition $Y$ from $Z \setminus Y$, meaning that they may partially overlap with regions $z \in Z$. It makes it more difficult to cleanly partition $Z$ into $Y$ and $Z \setminus Y$.

It is not immediately clear how to handle these in defining $b(Y)$ and $m(Y)$, at least not efficiently.  
One approach is the 
area-based framework of \cite{buchin2012}. They model the map as a polygonal subdivision with per-region case and population counts, and weight a regions contribution towards $\phi$ proportional to its overlap.  
Their method, however, only considers a fixed-scale polygonal window.   
They discretize the continuous placement problem by constructing the arrangement of combinatorially distinct placements of the polygonal shape with respect to the subdivision, and then optimize the likelihood within each cell.  
In experiments on real geography with synthetic cases, their area-based methods localized clusters more accurately than centroid-based baselines.  They also introduced a non-homogeneous variant that showed lower detection power. In practice, circular windows are approximated by regular polygons.  
A limitation of their method is the fixed scale and aspect ratio (for rectangles) of anomalous regions they consider.  

However, the most widely used heuristic to resolve these difficulties is to simply represent each region $z \in Z$ at a single point $x_z$ at its centroid, assigning $b(x_z) = b(z)$ and $m(x_z) = m(z)$. This allows for direct application of point-based scan algorithms. However, this approximation introduces geometric inaccuracies: a shape $S$ containing $x_z$ may only partially intersect $z$, leading to a possible misalignment between the shape and the underlying region. This approach \emph{ignores} this potential error.


\subsection{Software for Spatial Scan Statistics}
\label{sec:software}

Several software tools implement the spatial scan statistic, each with trade-offs on speed, flexibility, and accuracy in detecting regions.

Probably the oldest and most widely used implementation is SaTScan~\citep{Kul7.0}, which supports both point-based and region-based data.  For regions, it reduces the regions to single points—typically centroids—and scans over circular or elliptical windows to maximize $\Phi_{\SS}$.   It offers a variety of statistical models, as well as searching for space-time, purely temporal, and seasonally-defined anomalies.  

The pyScan~\citep{pyScan} offers a modern Python interface with significantly improved scalability. It relies on a well-engineered C++ backend with algorithmic enhancements for computing spatial scan statistics efficiently. Additionally, pyScan allows for guaranteed approximations of the optimal score, trading small controlled loss in accuracy for substantial gains in speed. In this work, we  enable a mild approximation in pyScan; see Section \ref{sec:exp}.  


Another option is FlexScan~\citep{FleXScan}, which operates directly on a region-based input.  It performs an exhaustive search over combinations of spatially connected regions to identify clusters.  FlexScan allows for more flexibility in cluster shape, but several previous studies have highlighted critical limitations.  \cite{tango2005flexibly} noted that FlexScan has high computational overhead, especially when applied to datasets with a large number of spatial regions; an efficient algorithm is needed.
Moreover, it is limited to clusters of at most 15 regions; if the anomalous area is larger than that, it reports multiple disjoint regions without indicating if they are contiguous or not.  In the large cluster case, it does not report a single region or single significance score.  
We also found that FlexScan requires more manual pre- and post-processing; for instance, if the dataset contained unconnected regions, it failed to report any clusters unless these were first removed.  



We also obtained a Java implementation (JDK~17) of the area-based method of \cite{buchin2012} by contacting the authors.
The method scans a window on a \emph{fixed} scale, as required by the code.  
Because the spatial extent of an anomaly is a key aspect of what one is trying to discover, this is a strong assumption.  
The paper recommends trying several scales (using $9$ multiples $\{0.5, 0.7, 0.85, 0.95, 1.0,$ $1.05, 1.15, 1.3, 1.5\}$ of the guess -- or given the true planted size) retaining the highest scoring anomaly across all scales, and we follow this in our experiments.  
The true size would not be known by a practitioner. By contrast, the scan methods of SaTScan and pyScan that our method extends are able to consider all scales from some shape family intrinsically.  

\section{Methods}

To address the challenge of converting region-aggregated spatial data to be compatible with point-based scan statistics, we propose a principled alternative to the common centroid approximation.

\paragraph*{Baseline: Centroid Method}  
In the centroid-based approach, each region $z$ is represented by a single point $x_z$ located at its geometric centroid. The entire measured and baseline values of the region are then assigned to that point, i.e. $m(x_z) = m(z)$ and $b(x_z) = b(z)$. Although computationally efficient, this reduction discards the internal spatial structure of the region, potentially diminishing statistical power and introducing bias.


\paragraph*{Our Proposed Sampling-Based Method}  
Instead of using one centroid point per region, we propose replacing each region $z$ with a set of $k$ representative points: $x_{z,1}, \ldots, x_{z,k}$. These points are sampled uniformly at random from within the geometry of $z$. The region’s values are then evenly distributed across these points, so that
\[
m(x_{z,j}) = \frac{m(z)}{k}, \quad b(x_{z,j}) = \frac{b(z)}{k}, \quad \text{for all } j = 1, \ldots, k.
\] 

This strategy preserves both the spatial extent and the internal structure of the original region while remaining compatible with existing point-based scan algorithms.





\begin{figure}[htbp]
  \includegraphics[width=0.9\linewidth]{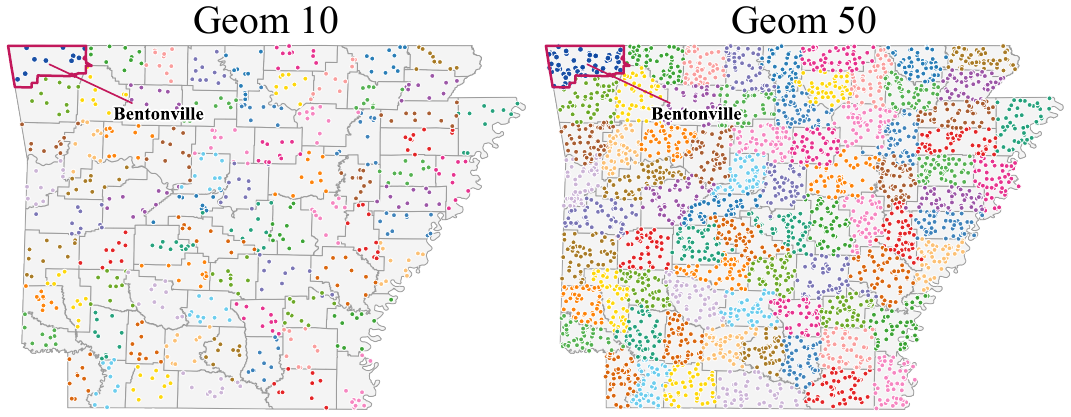}
  \Description{Two side-by-side maps of Arkansas with county polygon
boundaries drawn in light gray. Each county is filled with small
colored dots representing sampled points; counties are color-coded
distinctly so adjacent counties are visually separable. Benton
County in the northwest corner is outlined in magenta and labeled. The
left map (Geom 10) shows roughly ten dots per county; the right map
(Geom 50) shows roughly fifty dots per county at the same scale.}

  \vspace{-2mm}
  \caption{Region-to-point sampling on Arkansas counties: each county
  is replaced by $k$ points drawn uniformly at random from its
  polygon. For example Benton County (outlined), contains exactly $10$ points in
  \textsf{Geom 10} (left) and exactly $50$ points in \textsf{Geom 50}
  (right).}
  \label{fig:PointsamplingArkansas}
\end{figure}

While we explore other variants,
we find that random selection of these $k$ points in each region is effective. 
We explore the values of $k$ 
from $1$ (a single uniformly sampled point, i.e.\ \textsf{Random Point}; also the deterministic \textsf{Centroid}) to $50$.   Section~\ref{sec:k-choice} studies this choice empirically, finding diminishing returns beyond $k \approx 20$. 
Larger values of $k$ tend to produce higher statistical power, as they better represent the geometry of each region and, in the synthetic experiments, provide more observations of the planted signal. We find that while the results improve with larger $k$, the full $k=50$ points are not always needed. The optimal choice may depend on the complexity of the dataset and computational constraints. Figure~\ref{fig:PointsamplingArkansas} shows an example of this sampling-based conversion applied to counties in the state of Arkansas. The left panel uses $k = 10$ points per region, while the right panel uses $k = 50$.

\subsection{Recovering Planted Anomalies}


To evaluate the statistical power of the point-based methods, we consider their ability to recover anomalous regions that we synthetically plant in the data. We generate a planted shape $S\in\mathcal{S}$ and construct each point-based input: Centroid places one point at each region's geometric centroid, and Geom-$k$ samples locations uniformly within each region. Conditional on these locations, we independently include each point in the measured set with probability $p$ if inside $S$ and probability $q$ if outside it. We then pass the measured and baseline point sets to a scanning algorithm such as pyScan, which maximizes the Poisson-derived Kulldorff score, and assess recovery using the discovered region $\hat S$.
Our experiments primarily use planted rectangles and rectangular scan windows.

We evaluated how well algorithms work by how closely the identified anomalous region $\hat S$ (the \emph{Discovered Rectangle}) matches the planted one $S$ (the \emph{Target Rectangle}).  We do not expect these to match exactly due to spatial rounding errors and since each point is independently included in the measured set probabilistically.  As rates $p$ and $q$ become similar, the most anomalous rectangle from the generated data becomes less distinct from the generating one.  

To quantify the similarity between $S$ and $\hat S$, we use a modified Jaccard distance defined on a large set of fixed points $A$. 
We construct the fixed evaluation set $A$ by independently sampling $500$ points uniformly within each input region, so $|A|=500n$, where $n$ is the number of regions. Using a separate fixed seed, we build $A$ once per dataset and reuse it across all methods, values of $k$, rate contrasts, trials, and experiment seeds. This dense sampling ensures robust, reproducible overlap comparisons.

We define \emph{Point Jaccard distance} as the following:
\[
  d_{\textsc{Jac},A}(S, \hat S) = 1 - \frac{|A \cap (S \cap \hat S)|}{|A \cap (S \cup \hat S)|}.  
\] 
This distance is one minus the fraction of evaluation points in $S\cup\hat S$ that also lie in $S\cap\hat S$; lower values indicate greater overlap. 

\begin{figure}[b]
    \includegraphics[width=1.0\linewidth]{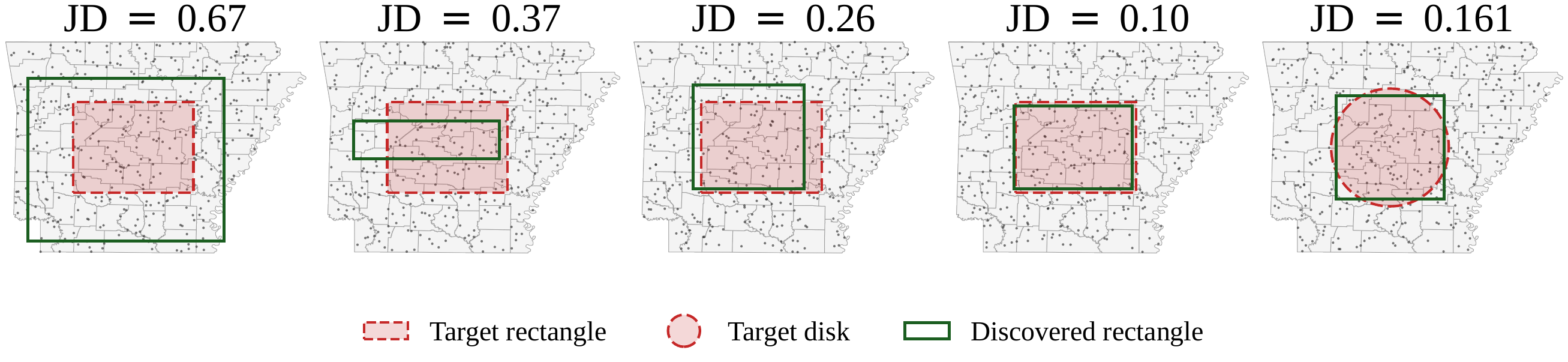}
    \Description{A horizontal row of five panels, each showing the outline
of Arkansas with county boundaries in light gray and small black dots
representing sampled points. Each panel is labeled at the top with a
Jaccard distance value: JD = 0.67, 0.37, 0.26, 0.10, and 0.161 from
left to right. The first four panels show a red dashed rectangle (the
planted target) overlapping a green solid rectangle (the discovered
rectangle); the two rectangles diverge more in the leftmost panel and
nest tightly in the fourth. The rightmost panel shows a red dashed
circle (planted disk) with a green solid rectangle (best-fit rectangle)
inscribed inside it.}

\vspace{-3mm}
    \caption{Point Jaccard distance between target (red) and discovered rectangles (green outline) for Arkansas counties. The last panel shows the best-fit discovered rectangle to a circular shape.}
    \label{fig:JDArkansasGeom5Rec}
\end{figure}

Figure \ref{fig:JDArkansasGeom5Rec} illustrates an example for Arkansas counties using the Geom 5 method. Lower Jaccard distance values indicate better recovery of the planted anomaly.
This provides visual intuition for the meaning of various Jaccard distances.  In particular, the last panel shows where the target region is a disk, not a rectangle; it has a Jaccard distance of $0.161$.  This is especially informative since in practice the true anomalies may not be rectangular, and this provides a sort of ``shape floor'' for how similar we should expect to recover a planted anomaly.  That is, below about $0.2$ Jaccard distance is probably sufficient.  On the other hand a Jaccard distance above $0.35$ is not a good fit.  


\section{Experimental Evaluation}
\label{sec:exp}
Our main experiment evaluates how effectively each scanning method recovers a planted anomalous region under varying conditions. We vary the algorithm's parameters and the experimental setup to see how algorithms perform in different case studies.
To capture a diverse range of spatial structures and region complexity, we conducted experiments on six different geospatial datasets. Arkansas, New York City, Utah, California, Georgia, and the entire United States. These cases represent a spectrum from local compact, densely populated regions (e.g., NYC) to large, sparsely populated ones (e.g., Utah).

Unless otherwise noted, each experimental result was produced by taking the average Jaccard distance of the $20$ trials.  Both the experimental setup and the newly proposed algorithms are random, so it is helpful to understand the distribution of the results.  We also plot shaded regions that show the results' average values plus or minus one standard deviation.  Methods with narrower bands are more consistent, whereas wider bands indicate a higher sensitivity to randomness in the setup. 

\paragraph*{Measuring Statistical Power}
While there are various ways to formulate it, \emph{statistical power} for a method informally captures how reliably it can solve a task; in our case the probability of detecting a planted anomaly at a certain difficulty level. 
In the classical sense, power is the probability of rejecting the null hypothesis of no cluster when an anomaly is truly present, which presumes a significance test at a fixed level~\citep{glaz2024handbook,tango2005flexibly}. 
We instead report how closely each method \emph{localizes} a planted region, via Point Jaccard distance averaged over trials, and treat a method as having higher power when it recovers the region at smaller $pq$ differences. The two notions are related but not interchangeable: a method could reject the null and still mislocate the cluster.

We vary problem difficulty, and measure the statistical power, using two factors: the \textit{$pq$ difference} and the \textit{size of the planted region}. The $pq$ difference (our primary focus) captures the contrast between the measured rates inside ($p$) and outside ($q$) the planted region. 
In our experiments, we always set $q=0.2$, so for each baseline value $b(z)$ outside a planted region, we expect to observe $\mathsf{E}[m(z)] = (0.2) b(z)$ as the corresponding measured value.  We vary the rate within the region from $p=0.2$ (so the $pq$ difference is $0$) and $p=0.9$ (so the $pq$ difference is $0.7$).  
The more different these values, the more apparent the anomaly should be, and the easier it should be to detect reliably.  When $p = 0.9$, the difference ($p - q = 0.7$) should present an obvious anomaly; when $p = 0.2$, the planted region should be indistinguishable from the background.

A method with large statistical power can with high probability recover the planted region even at lower $pq$ differences. Thus, we interpret the lower thresholds of $p$ (with fixed $q$) at which recovery is still reliably successful as clear indicators of higher power.


\paragraph*{Methods Considered}
We evaluated five point-based conversion strategies using the \texttt{pyScan} (with adaptive grid at size $100\times100$)  
and restrict our scan shapes to rectangles.  We chose this shape because it provides the best combination of representational complexity and computational efficiency.  
We consider a region-based input and convert the problem to point-based input in $5$ different ways.  
\textsf{Centroid} is the main baseline approach~\cite{Kul7.0}, where each region is represented by a single point on its geometric centroid.  
\textsf{Random Point} replaces each region with one randomly sampled from within the region's shape.  This is a variation of our proposed method but uses only a single point.  This variant shares the same computational cost as the centroid method, but introduces randomness.  
Later we more carefully compare to FlexScan~\cite{tango2005flexibly,tango2012restrict} and Buchin \emph{et al.}~\cite{buchin2012}.  

\textsf{Geom 5}, \textsf{Geom 10}, and \textsf{Geom 50} represent our proposed multi-point sampling strategy. They replace each region with $5$, $10$, or $50$ randomly sampled points, respectively, distributing the region’s baseline and measured values evenly across those points.


\subsection{Zip Codes in New York City}
\label{sec:nyc}

We start our evaluation by studying the zip code boundaries of New York City (263 polygons representing 248 unique ZIPs; ten ZIPs, e.g., 10004, are split into multiple polygons by water, and we treat each polygon as an independent region in the scan).
These shape files vary  significantly in size and geometry -- although they are designed to represent zip codes with approximately equal populations.  We assign the same baseline value to each zipcode region.    
We synthetically plant a rectangular region in a central portion of the map, covering approximately one-third of the total area. This planted region is shown inset in Figure~\ref{fig:NewYorkCityTRPlot}.  Using this shape, we fix $q=0.2$, and vary the rate $p$ inside the shape.  

Figure \ref{fig:NewYorkCityTRPlot} shows the Point Jaccard distance as a function of the $pq$ difference on the $x$-axis for each method. Both the Centroid and Random Point approaches exhibit similar behavior: they are relatively noisy and consistently identify the planted region only when the $pq$ difference is approximately $0.35$ or $0.4$.  In contrast, as we increase the number of random points from $1$ to $5$, $10$, or $50$ per region, the results become more stable and the planted region is detected with much smaller $pq$ differences.  

\begin{figure}[t]
    \includegraphics[width=0.95\linewidth]{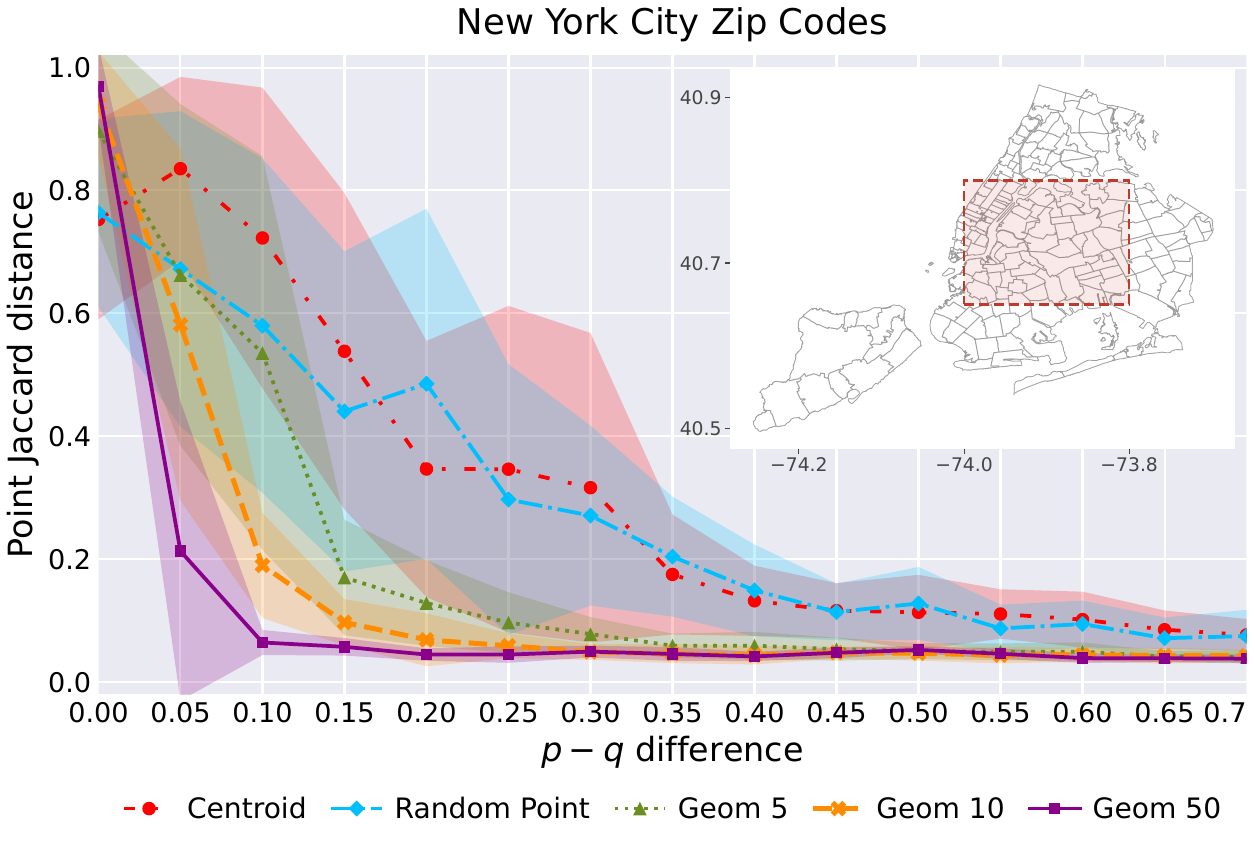}
    \Description{A line plot of Point Jaccard distance (y-axis, 0 to 1)
versus p minus q difference (x-axis, 0 to 0.7) with five curves, one
per method: Centroid, Random Point, Geom 5, Geom 10, and Geom 50, each
in a distinct color with shaded standard-deviation bands. The curves
generally decrease as $p-q$ increases. Geom 50 is approximately $0.213$
at $p-q=0.05$ and falls below $0.1$ at $p-q=0.10$; Centroid and Random
Point decrease more slowly and have wider standard-deviation bands.
In the upper-right corner an inset map shows the outline of
New York City zip-code boundaries in light gray with a small red dashed
rectangle marking the planted target region.}
    \caption{New York City with $263$ polygons representing $248$ unique ZIP codes. The inset (upper right)
    shows the planted target rectangle (red dashed). The 
    Point Jaccard distance versus $pq$ difference, averaged over $20$ trials with $\pm1$ standard
    deviation bands.}
    \label{fig:NewYorkCityTRPlot}
\end{figure}

Specifically, with $5$ random points, the planted region is reliably recovered at $pq$ differences of $0.2$ or greater; with $10$ points, detection occurs at $pq$ differences as low as $0.1$. With $50$ points per region, recovery is close to the $0.2$ threshold at $p-q=0.05$, where Point Jaccard distance is $0.213$, and is strong by $p-q=0.10$, where the distance falls to $0.064$. 


Note that even when the $pq$ difference is very large (e.g., at least $0.5$), the \textsf{Geom} methods converge to a Point Jaccard distance of around $0.05$, and Centroid and Random Point to around $0.1$.  
Although we could potentially devise a way to measure the accuracy of the region recovery so that it converges to $0$ error, this nonzero plateau is consistent with randomness in the synthetic data-generation process.

\subsection{Counties of US States}

\paragraph{Utah.}
Next, we consider the 29 counties in the state of Utah; see the inset panel of Figure \ref{fig:UtahTRPlot}.  This example has few counties which are often irregularly shaped, which introduces additional geometric challenges for region-based analysis.


\begin{figure}[t]
    \includegraphics[width=0.95\linewidth]{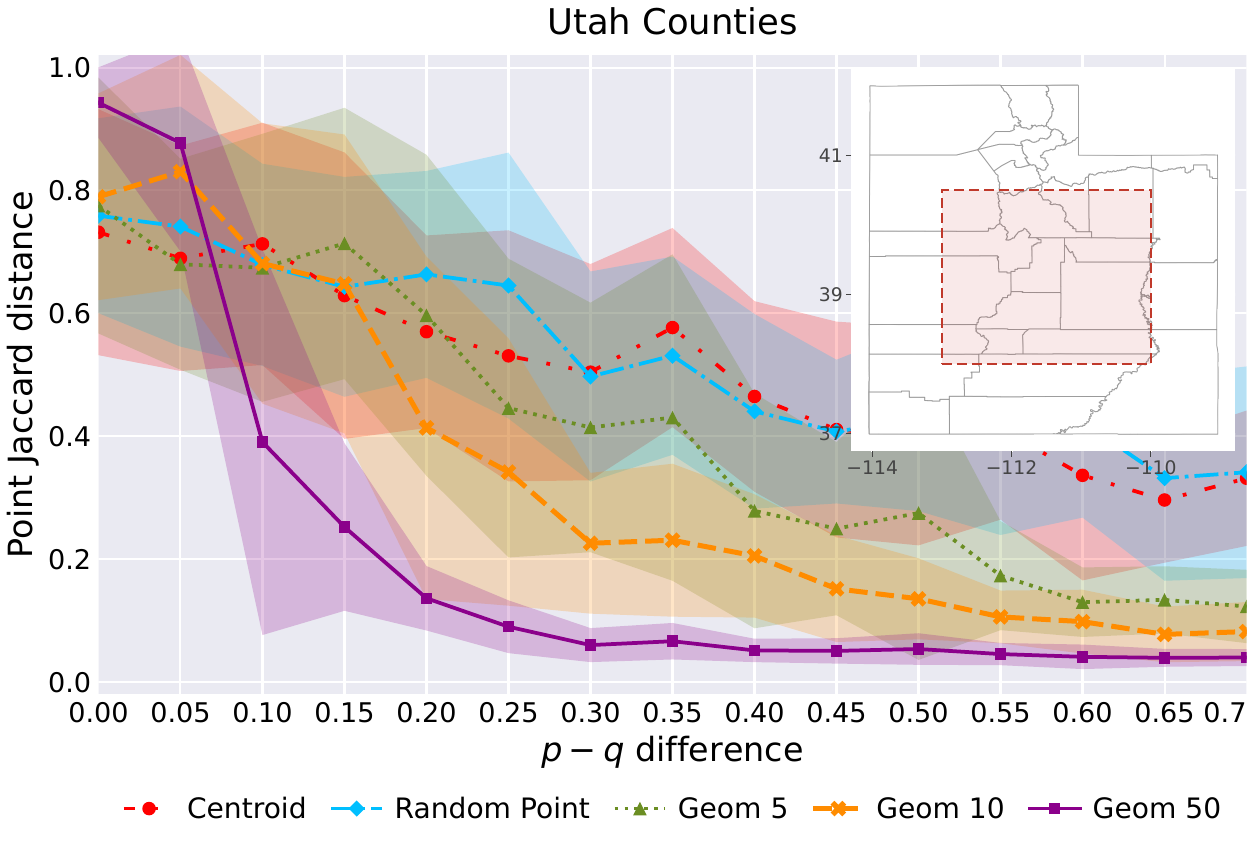}
    \Description{A line plot of Point Jaccard distance (y-axis, 0 to 1)
versus p minus q difference (x-axis, 0 to 0.7) with five curves
representing the methods Centroid, Random Point, Geom 5, Geom 10, and
Geom 50, each in a distinct color with shaded standard-deviation bands.
The curves generally decline as $p-q$ increases, with Geom 50 lowest at
the right and Centroid and Random Point remaining near $0.33$--$0.34$.
The Centroid and Random Point curves are visibly noisier
with wide bands, while Geom 50 drops fastest and most smoothly. In the
upper-right corner an inset map shows the outline of Utah and its 29
counties in light gray with a red dashed rectangle marking the planted
target region covering most of the eastern half of the state.}

    \vspace{-3mm}
    \caption{Utah with 29 counties. The inset (upper right) shows the
    state and the planted target rectangle (red dashed). The Point Jaccard distance versus $pq$ difference.} 
    \label{fig:UtahTRPlot}
\end{figure}

Figure~\ref{fig:UtahTRPlot} shows, due to the number of regions and their complex shapes, substantial variability in the Point Jaccard distance for both the Centroid and Random Point approaches, as well as for Geom~5. These methods struggle to consistently detect the planted region, particularly at low $pq$ differences.

The Geom~10 and Geom~50 approaches show more reliable behavior. Although some instability persists at small $pq$ differences, their performance stabilizes as the $pq$ difference increases. 


\begin{figure}[b]
    \includegraphics[width=0.9\linewidth]{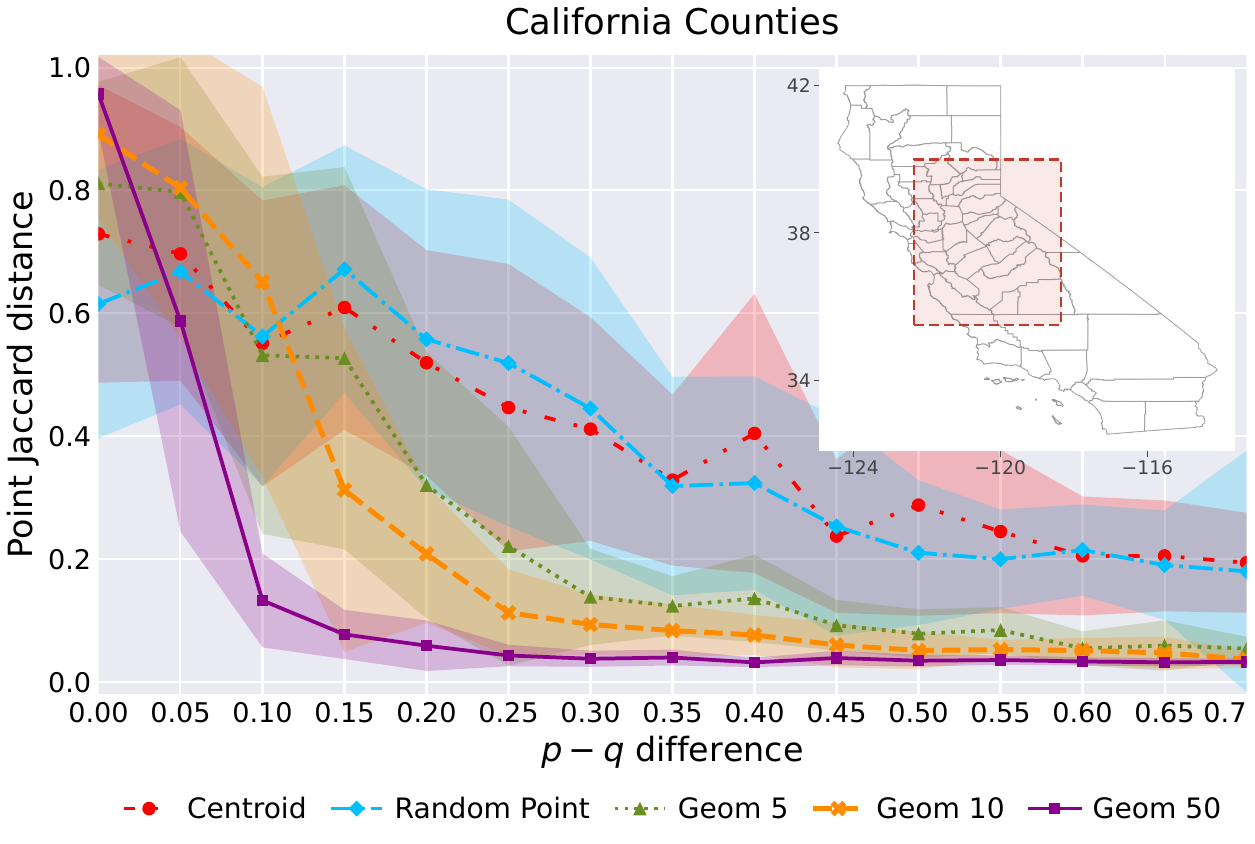}
\Description{A line plot of Point Jaccard distance (y-axis, 0 to 1)
versus p minus q difference (x-axis, 0 to 0.7) with five curves
representing the methods Centroid, Random Point, Geom 5, Geom 10, and
Geom 50, each in a distinct color with shaded standard-deviation bands.
The curves generally decrease as $p-q$ increases; at the right, Geom 50
is approximately $0.03$, while Centroid and Random Point are
approximately $0.2$. The Centroid and Random Point curves have wide bands indicating
high variability, while Geom 50 drops fastest and most smoothly. In the
upper-right corner an inset map shows the oblong outline of California
with its 58 counties (69 polygons)
in light gray and a red dashed rectangle marking
the planted target region overlapping the Bay Area and extending beyond
the state's western boundary.}

    \vspace{-3mm}
    \caption{California, 
    69 polygons over 58 counties.
    The inset
    (upper right) shows the state and the planted target rectangle (red
    dashed).
    The main plot shows Point Jaccard
    distance.} 
    \label{fig:CaliforniaTRPlot}
\end{figure}

\paragraph{California.}
The counties in the state of California present a distinct challenge for spatial analysis. Although there are 
58 counties\footnote{California has 58 counties, however our shapefile splits three coastal counties by their offshore islands leading to 69 polygons, which we treat as independent regions.},
the state's shape is oblong and the counties are of vastly different sizes, introducing spatial heterogeneity and complicating detection.  The scenario we consider, shown in Figure \ref{fig:CaliforniaTRPlot}, has the target rectangle that separates several smaller counties in the Bay Area and even includes areas outside the state boundaries.

\begin{figure}[t]
    \includegraphics[width=0.95\linewidth]{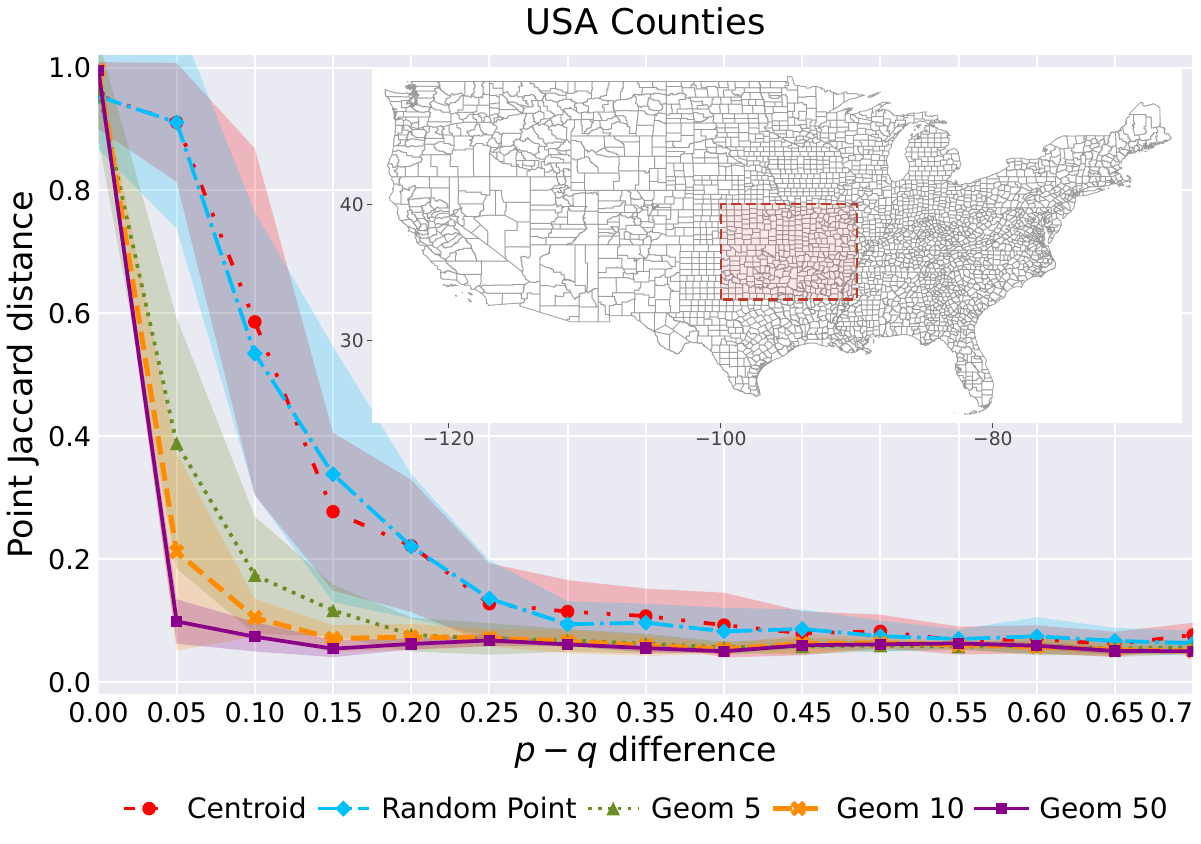}
    \Description{A line plot of Point Jaccard distance (y-axis, 0 to 1)
versus p minus q difference (x-axis, 0 to 0.7) with five curves
representing the methods Centroid, Random Point, Geom 5, Geom 10, and
Geom 50, each in a distinct color with shaded standard-deviation bands.
All curves decline from near 1 at the left toward near 0 at the right,
and the bands are narrower than in smaller-region datasets, indicating
more stable behavior. Geom 50 reaches Point Jaccard distance approximately $0.10$ at
$p-q=0.05$, while Geom 10 falls below the $0.20$ recovery threshold at
$p-q=0.10$. In the upper-right corner an
inset map shows the outline of the contiguous United States with all
3{,}108 counties in light gray and a red dashed rectangle marking the
planted target region covering parts of the central and southern US.}

    \vspace{-3mm}
    \caption{Point Jaccard distance for $3{,}108$ counties in the
    contiguous United States. The inset shows the country and the planted
    target rectangle.}
    \label{fig:USATRPlot}
\end{figure}

The results highlight the effectiveness of our proposed method. Geom~50 consistently achieves lower Point Jaccard distances across a range of $pq$ differences, indicating high statistical power. While the other methods, including the centroid approach, can achieve a small Point Jaccard distance, they have much more variation in what they find, demonstrated by much larger standard deviation.



\paragraph{Contiguous USA.}
We evaluate our methods on the $3{,}108$ counties in the contiguous
United States; Figure~\ref{fig:USATRPlot} shows the results.
The larger number of regions and the increased spatial coverage create a more stable statistical environment, allowing all methods to perform more reliably.

As shown in Figure~\ref{fig:USATRPlot}, all methods benefit from the larger number of regions, but the Centroid method improves more slowly and remains above $0.1$ at a $pq$ difference of $0.3$. The proposed methods exhibit more statistical power: Geom~50 reaches approximately $0.10$ at a $pq$ difference of $0.05$, while Geom~10 is approximately $0.21$ at $0.05$ and approximately $0.10$ at $0.1$.


\begin{figure}[htbp]
    \includegraphics[width=0.95\linewidth]{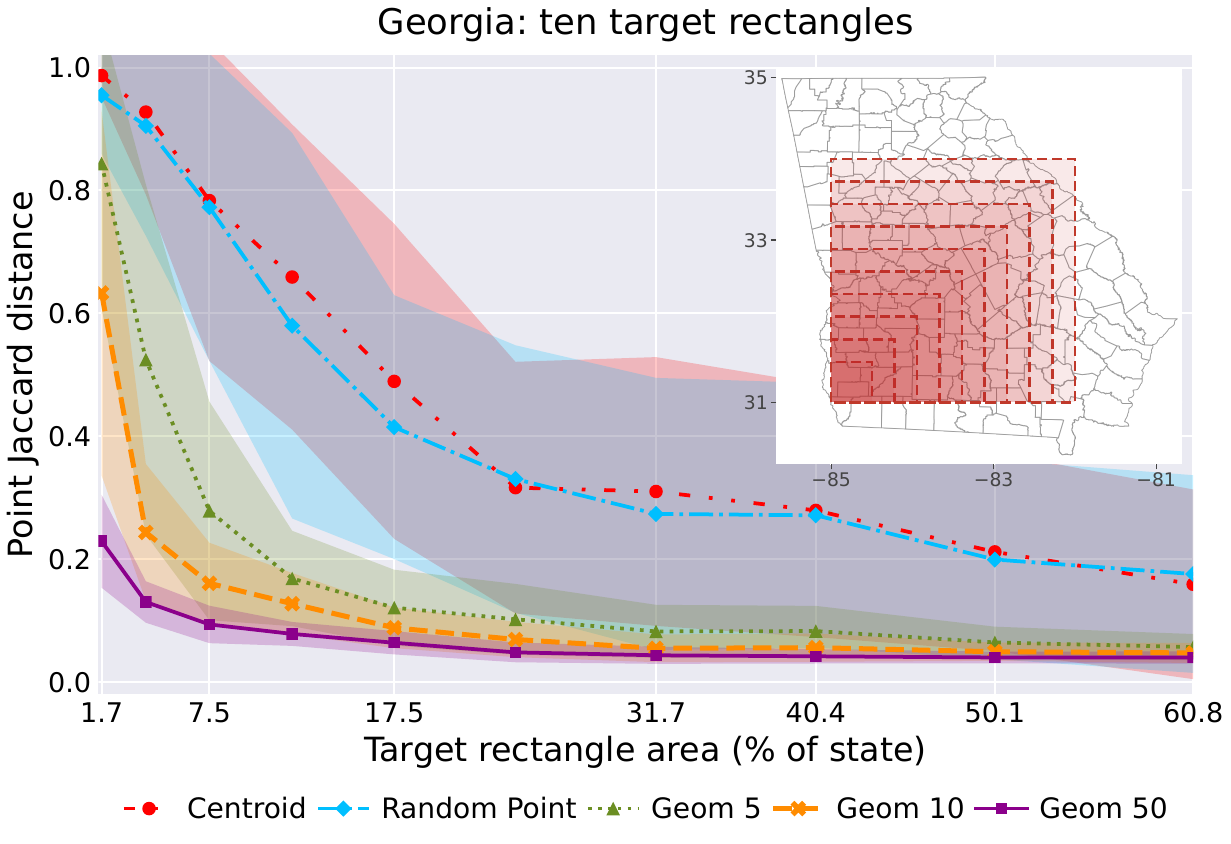}
    \Description{A line plot of Point Jaccard distance (y-axis, 0 to 1)
versus target rectangle area as a percentage of the state of Georgia
(x-axis, with tick labels 1.7, 7.5, 17.5, 31.7, 40.4, 50.1, and 60.8)
with five curves representing the methods Centroid, Random Point, Geom
5, Geom 10, and Geom 50, each in a distinct color with shaded
standard-deviation bands. The Centroid and Random Point curves stay
high (above 0.6) for small target sizes and drop slowly; Geom 50 falls
below 0.2 for targets covering approximately $4.1\%$ of the state, while
Geom 10 and Geom 5 require larger targets. In the upper-right corner an inset map shows the outline of Georgia with its counties in light gray and ten nested red dashed rectangles of
increasing size centered on roughly the same point.}
    \caption{Point Jaccard distance as a function of the relative size of the planted target rectangle in Georgia, with $p-q$ fixed at $0.4$ and averaged over $80$ trials. The inset shows the ten nested target rectangles overlaid on Georgia counties.}
    \label{fig:GeorgiaRecs}
\end{figure}

\subsection{Effect of Target Rectangle Size}
To further evaluate robustness, we conducted a different experiment to show how the methods are affected by the size of the planted region. We started by selecting a small rectangle and progressively increasing its size from $2\%$ to $61\%$ of the state of Georgia; see Figure \ref{fig:GeorgiaRecs}. In this experiment, we fixed the $pq$ difference at $0.4$ to isolate the effect of the region size on the statistical power.  

As the target rectangle becomes smaller, it contains fewer anomalous signals, making it more difficult to distinguish it from natural random variation elsewhere.  The results in Figure~\ref{fig:GeorgiaRecs} show that the Centroid and Random Point methods struggle across all region sizes, with performance degrading substantially for regions smaller than $12\%$ of the state. 



In contrast, our proposed methods recover substantially
smaller targets. At the $0.2$ recovery threshold, Geom 5, Geom 10, and
Geom 50 recover targets covering as little as approximately $12.0\%$,
$7.5\%$, and $4.1\%$ of the state, respectively. By comparison, Random
Point first reaches the threshold at approximately $50.1\%$, while
Centroid reaches it only at the largest tested target size, $60.8\%$.
This highlights that adding even a few random points (instead of 1) in each region can greatly enhance the statistical power of the methods to detect significantly small spatial anomalies.

\subsection{Comparison with FlexScan and Buchin et al.}
\label{sec:flexscan-buchin}
We also evaluate our proposed method in counties in the state of Arkansas and compare its performance with the FlexScan\footnote{FleXScan v3.1.2 software manual instructions \citep{FleXScan}} technique (a brown curve), and the area-based method of Buchin et al.~\citep{buchin2012} (a dark blue curve). FlexScan identifies a connected region of counties with elevated rates.  
The Arkansas dataset includes a moderate number of counties (75), each with fairly uniform shapes, making it a favorable setting for FlexScan.  
As in previous experiments, we plant a rectangular anomaly, this time covering approximately 30\% of the state; see Figure~\ref{fig:ArkansasFlexScan30Plot}.  

\begin{figure}[htbp]
    \includegraphics[width=0.9\linewidth]{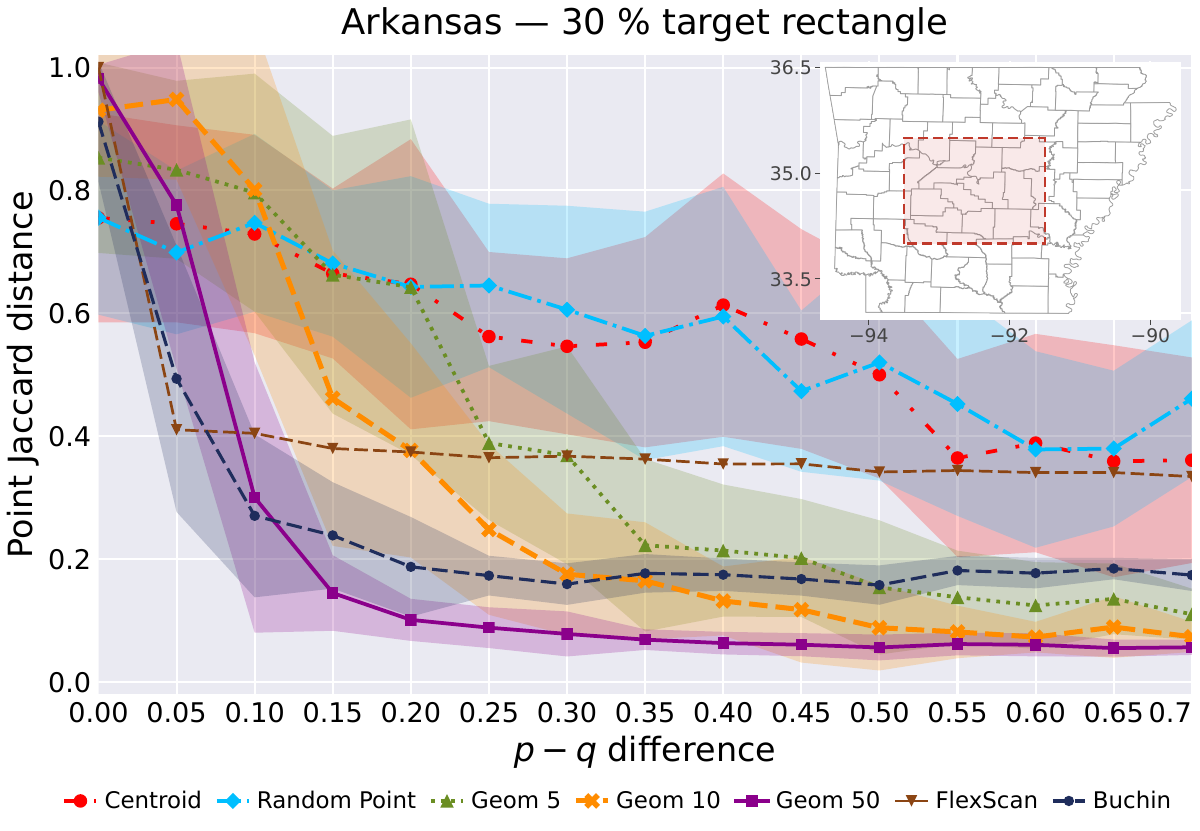}
    \Description{A line plot of Point Jaccard distance (y-axis, 0 to 1)
versus p minus q difference (x-axis, 0 to 0.7) with seven curves
representing the methods Centroid, Random Point, Geom 5, Geom 10, Geom
50, FlexScan, and Buchin, each in a distinct color with shaded
standard-deviation bands. The Centroid and Random Point curves stay
noisy and high; FlexScan plateaus around 0.35 to 0.4 even as the x-axis
increases; the Buchin and Geom 50 curves drop fastest and reach the
lowest values, with Geom 50 slightly below Buchin at the right end. In
the upper-right corner an inset map shows the outline of Arkansas with
counties in light gray and a red dashed rectangle marking the planted
target region covering roughly 30 percent of the state.}
    \vspace{-3mm}
    \caption{Point Jaccard distance between the planted target and the discovered region for Arkansas, with the target rectangle covering about $30\%$ of the state.}
    \label{fig:ArkansasFlexScan30Plot}
\end{figure}

\begin{figure}[t]
    \includegraphics[width=0.9\linewidth]{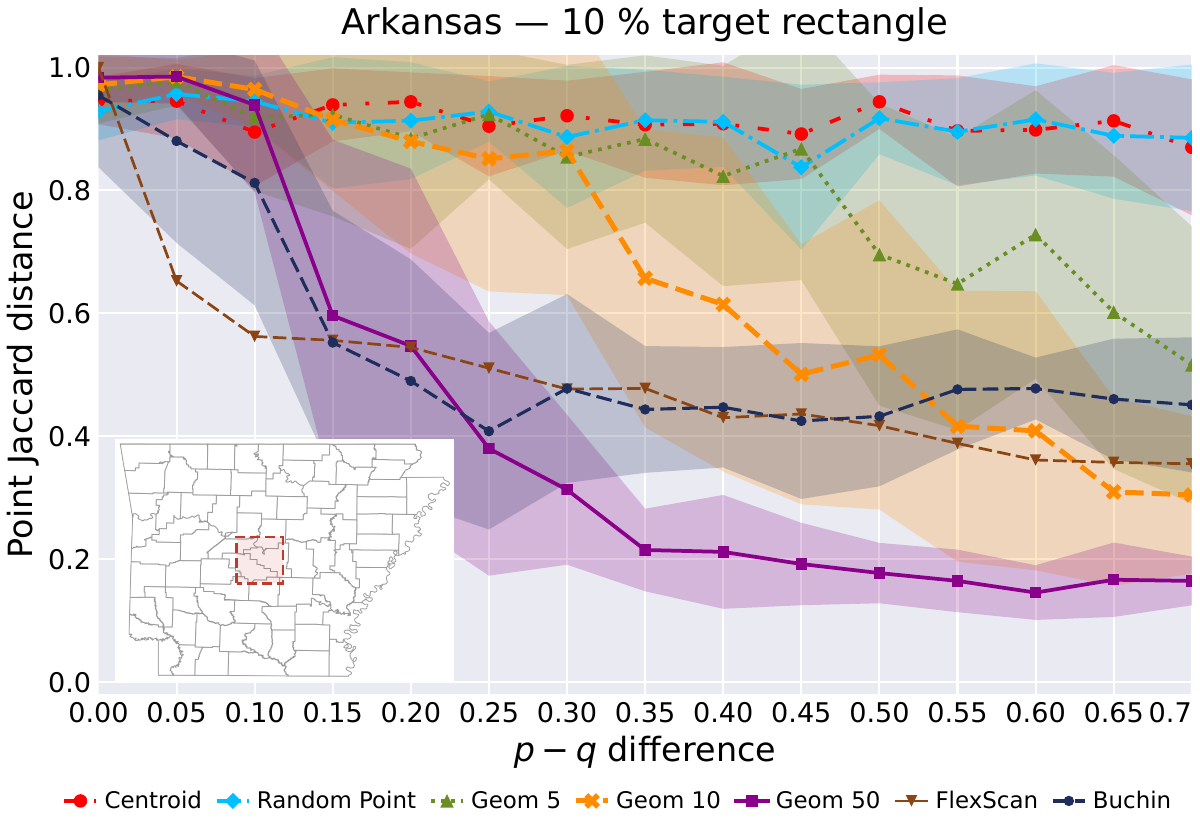}
    \Description{A line plot of Point Jaccard distance (y-axis, 0 to 1)
versus p minus q difference (x-axis, 0 to 0.7) with seven curves
representing the methods Centroid, Random Point, Geom 5, Geom 10, Geom
50, FlexScan, and Buchin, each in a distinct color with shaded
standard-deviation bands. Compared to the larger-target case, the
curves separate more sharply: Centroid and Random Point stay near 1
across most of the x-axis; FlexScan plateaus around 0.35 to 0.4; Buchin
plateaus around 0.45; and only Geom 50 drops cleanly below 0.2, doing
so by p minus q of about 0.25. In the lower-left corner an inset map
shows the outline of Arkansas with counties in light gray and a small
red dashed rectangle marking the planted target region covering roughly
10 percent of the state.}
    \vspace{-3mm}
    \caption{Point Jaccard distance between the planted and discovered regions for Arkansas, with the target rectangle covering about $10\%$ of the state.}
    \label{fig:ArkansasFlexScan10Plot}
\end{figure}

The plot in Figure \ref{fig:ArkansasFlexScan30Plot} shows how the methods perform as the $pq$ difference increases. We observe that the performance of the Centroid, Random Point, and Geom 5, 10, 50 methods follows trends similar to those of previous experiments. 

\paragraph{FlexScan}
While FlexScan performs reasonably, its Point Jaccard distance plateaus around $0.35$–$0.4$ 
(above the $0.2$ practical recovery threshold), 
even as the $pq$ difference increases.  This appears to be an artificial implementation limitation, as the code limits the search for the potentially anomalous cluster to at most 15 regions.

To ensure a fair comparison, we repeated the experiment on the Arkansas dataset using a smaller target rectangle covering only 10\% of the area, results shown in Figure~\ref{fig:ArkansasFlexScan10Plot}.  This setup is more challenging for both the Centroid and our proposed methods. Nevertheless, Geom 50 remains the best performer, achieving a Point Jaccard distance of approximately 0.2 (or less) for $pq \geq 0.35$.
FlexScan, on the other hand, fails to achieve a Point Jaccard distance below $0.35$, even as $pq$ increases.  Notably, for very small $pq$ differences ($\leq 0.1$), FlexScan outperforms our methods. Still, this advantage is limited, as the Point Jaccard distance is still quite high (Jaccard distance $\approx 0.6$), indicating that the overlap between the target and the discovered regions has some overlap on average, but not much.

\textit{Buchin.}  On the larger target (Figure~\ref{fig:ArkansasFlexScan30Plot}) the area-based method of Buchin et al.~\citep{buchin2012} tracks \textsf{Geom 50} closely but with on average a point Jaccard $0.1$ larger, and selecting nearly the correct scale (on average $1.07\times$ the planted area; offset $\sim\!6.0$ km) even though it was given the planted size.  On the smaller target (Figure~\ref{fig:ArkansasFlexScan10Plot}) the two methods have a larger gap: Buchin systematically oversizes to on average $1.61\times$ planted (offset $\sim\!11.4$ km) and plateaus near $0.45$, while \textsf{Geom 50}, whose continuous search adapts the scale, stays at $0.95\times$ planted (offset $\sim\!3.1$ km) and reaches $0.16$ for Point Jaccard.  
Figure~\ref{fig:RectMapAllMethods} shows these discovered rectangles geographically.  
Unlike Buchin \emph{et al.}'s area-based method, our approach requires no user-supplied window size and is roughly $90\times$ faster on Arkansas (Section~\ref{sec:runtime}, Table~\ref{tab:runtime_usa+AR}).  A comparison on \emph{disk-shaped} planted targets, where pyScan can match the target exactly, is reported in Appendix~\ref{app:disk-buchin}.

\begin{figure}[t]
  \centering
  \includegraphics[width=0.70\linewidth]{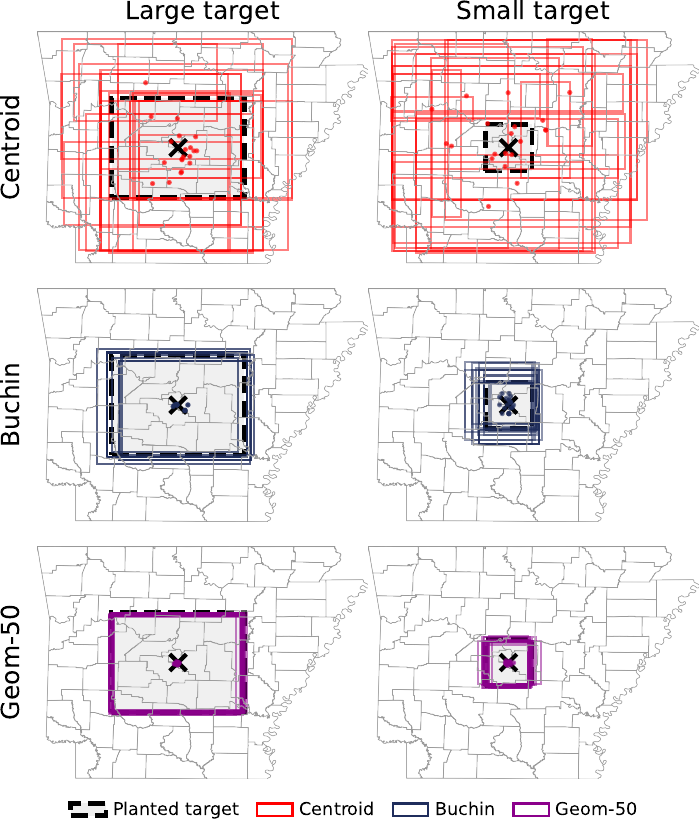}
  \Description{A 3-by-2 grid of maps of Arkansas with county boundaries
drawn in light gray. Each map has a black dashed rectangle marking the
planted target and a small black x at its center. The rows are labeled
Centroid, Buchin, and Geom-50; the columns are labeled Large target and
Small target. Each map shows about twenty overlapping rectangles in
that row's color: red for Centroid, navy for Buchin, magenta for
Geom-50. In the top row (Centroid) the rectangles spread widely beyond
the target on the large target and are dramatically oversized and
displaced on the small target. In the middle row (Buchin) the
rectangles nest closely around the large target and form a thick ring
slightly larger than the small target. In the bottom row (Geom-50) the
rectangles cluster tightly on both planted targets.}
  \caption{Discovered rectangular windows on Arkansas for 
  planted targets (black dashed) at $p-q = 0.5$, $20$ trials each.  
  \textbf{Rows:} \textbf{Centroid} (red),
  \textbf{Buchin} (navy), and \textsf{Geom 50} (magenta).}
  \label{fig:RectMapAllMethods}
\end{figure}

\subsection{Ablation on Sampling and Similarities}
\label{ablation}

Finally, we consider two design choices: (1) the similarity metric used to compare the target and discovered regions, and (2) the sampling points strategy within each region.  Our findings indicate that, while these choices introduce slight differences in performance, the selected default settings yield consistently strong results. 

\begin{figure}[htbp]
    \includegraphics[width=1.02\linewidth]{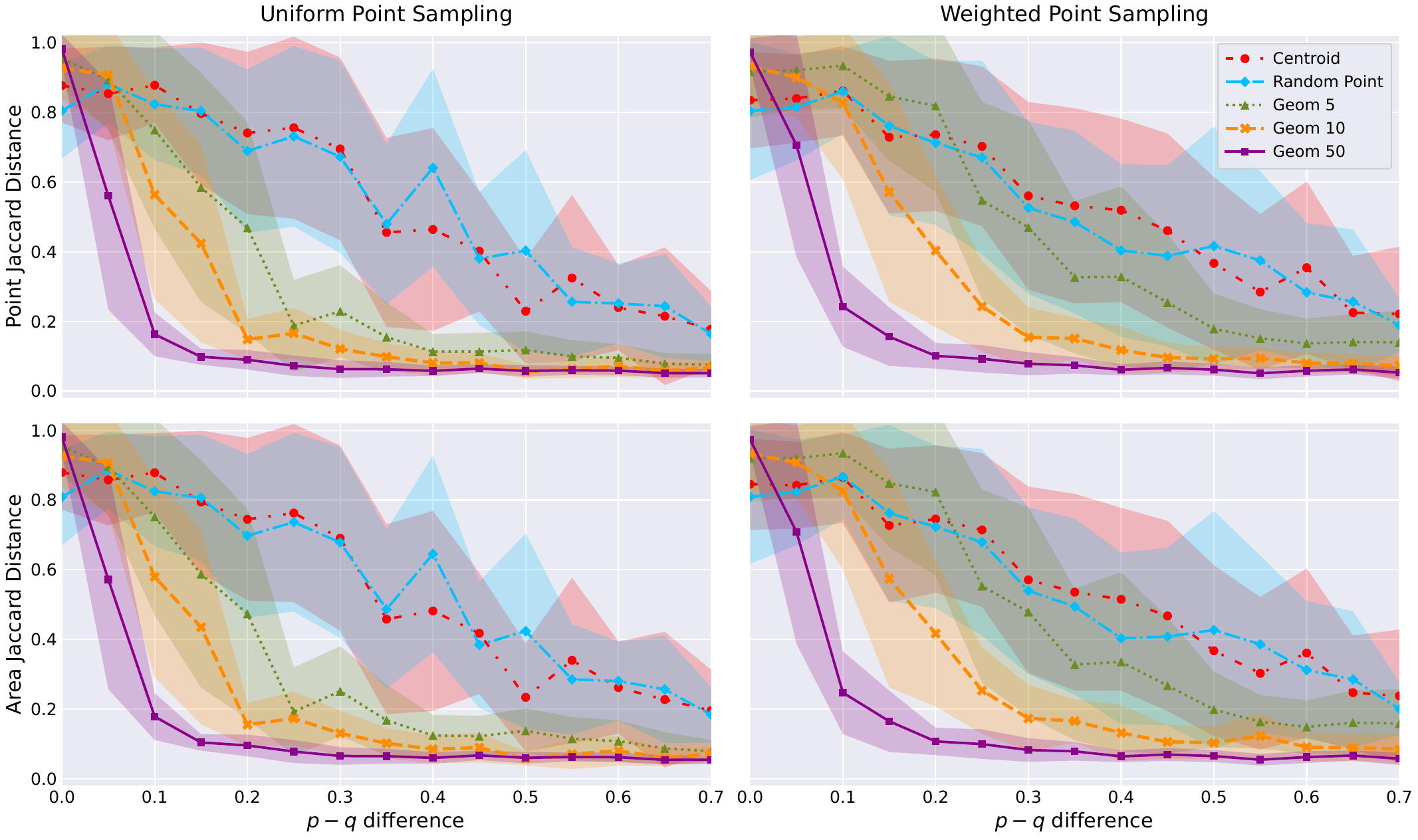}\\[0.6em]
    \includegraphics[width=1.02\linewidth]{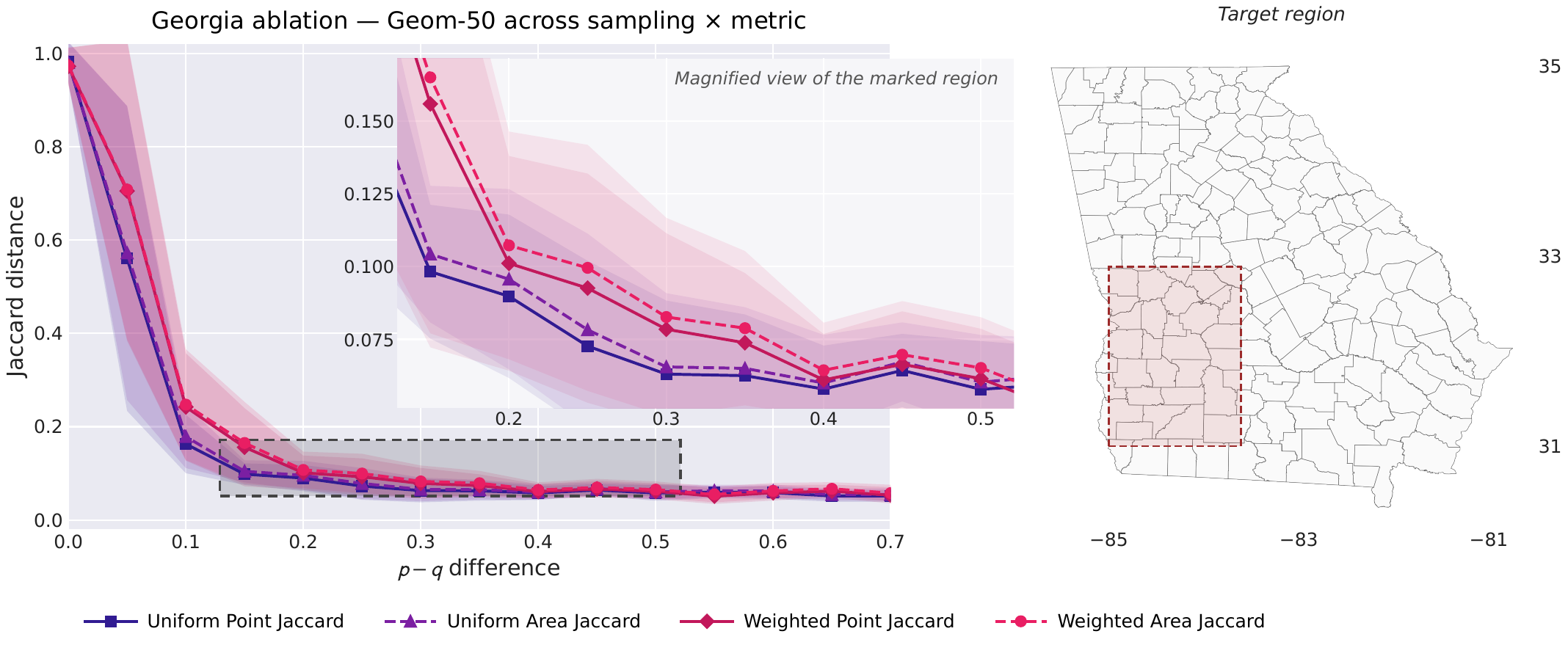}
    \Description{The figure has two stacked image groups. The top group is
a 2 by 2 grid of line plots, all showing distance versus p minus q
difference (x-axis 0 to 0.7) with five curves per panel for the methods
Centroid, Random Point, Geom 5, Geom 10, and Geom 50, each with shaded
standard-deviation bands. The top row uses Point Jaccard distance on
the y-axis; the bottom row uses Area Jaccard distance. The left column
uses Uniform Point Sampling; the right column uses capped Weighted Point
Sampling. Across all four panels the Geom curves drop fastest, with
Centroid and Random Point staying noisier and higher. The bottom group
has two side-by-side panels: a line plot comparing Geom 50 across four
settings (uniform-point, uniform-area, weighted-point, weighted-area)
with a magnified inset zooming into p minus q between 0.15 and 0.50
where the four curves separate, and a small map of Georgia counties on
the right with a red dashed rectangle marking the planted target
region.}

    \vspace{-2mm}
    \caption{Georgia counties with
    a planted region covering approximately $30\%$ of the state.
    \textbf{Top:} the $2\times2$ grid shows Point Jaccard distance (top row) and Area Jaccard distance (bottom row) under uniform point sampling
    (left column) and capped weighted point sampling (right column) across the five methods. \textbf{Bottom:} focused comparison of Geom~$50$ across the four (sampling, metric) settings, with the magnified panel zooming into
    $p-q\in[0.15,\,0.50]$.}
    \label{fig:Georgia4Plots}
\end{figure}

Our main evaluation metric, the \emph{Point Jaccard distance}, is based on comparing point sets inside each region. Specifically, we sample $500n$ points from the $n$ regions and compute the Jaccard distance between the sets associated with the target and discovered regions.

As an alternative, we could consider the \emph{Area Jaccard distance}, which compares regions directly via geometric overlap. This metric is defined as one minus the ratio of the intersection area to the union area between the target and discovered regions. It provides a shape-based similarity score.


We also compare two strategies for sampling points within each region:
\emph{Uniform Sampling:} Each region receives an equal number of sampled points (e.g., 1, 5, 10, or 50).
\emph{Weighted Sampling:} Each region receives a method-specific minimum number of points, with more points assigned to more populous regions, up to an upper cap. This aims to reflect a higher signal strength in denser regions.

While weighted sampling might seem more natural, it can lead to poor performance in sparse regions. For example, in the Georgia dataset, densely populated areas such as Atlanta are capped, while smaller or rural counties receive few points—diminishing the ability to detect target rectangles that span low-density areas.



Figure \ref{fig:Georgia4Plots} shows all four combinations in Georgia counties, using Geom 50 on a single plot using a target rectangle covering about $30\%$ of the area. It again shows the two uniform sampling 
settings performing slightly better for small values of the $pq$ difference, though the two arms do not use the same number of points.

Among the top 4 plots, the first row uses the Point Jaccard distance and the second row uses the Area Jaccard distance.  The first column uses uniform point sampling, and the second uses capped weighted point sampling.  While there are some differences among these four plots, all give the same general effect.  Perhaps the only significant signal is that under uniform point sampling, Geom 5 and Geom 10 more consistently recover the target region for $pq$ difference in the range of $0.15$ to $0.25$, when compared to using weighted sampling.

\section{Runtime and Convergence}

\label{sec:runtime}

Converting a problem with $n$ input points to $50 n$ (as we did in Geom 50) points could lead to a significant computational slowdown, changing problems from tractable to intractable.  Although naive algorithms for computing spatial scan statistics over rectangles on $n$ points scale with $O(n^4)$ or $O(n^5)$, there are improvements that reduce this to near quadratic in $n$; Agarwal \emph{et al.}~\citep{APV06} show how to compute this in $O(n^2 \log n)$ time.  An efficient version of this algorithm is implemented in pyScan, which we employ in this paper.  Even with this optimized implementation, a $50\times$ increase in data could imply a theoretical $2{,}500\times$ slowdown.   
Moreover, we leverage mild approximation in pyScan setting its \texttt{max\_subgrid} to an adaptive resolution grid at $100 \times
100$.  Its running time now partially depends on this grid size $g=100$ as $O(g^2 \log g + n)$~\cite{MP18}. 
Increasing the points per region therefore has a more limited
effect on runtime. Concretely, going from Arkansas ($n=75$) to
the USA ($n=3{,}711$ county-by-congressional-district regions) is a $50\times$ increase in regions, yet the scan
time grows only about $1.2\times$ (Table~\ref{tab:runtime_usa+AR}).

All runtimes were measured as wall-clock time on an Apple M2 Pro (8 performance + 4 efficiency cores), 16\,GB RAM, macOS, Python 3.10. 
They consider a single region discovery
at $p = 0.6$, $q = 0.2$, measuring only the scan step (case generation
and data loading excluded).  Buchin \emph{et al.}, as default, considers $9$ sizes chosen around the planted size.  FlexScan's number
additionally includes a Monte Carlo significance test that the other
methods do not perform; but we were unable to decouple in running the code.  


\begin{table}[t]
    \setlength{\tabcolsep}{3pt}
    \small
    \begin{tabular*}{\linewidth}{@{\extracolsep{\fill}}rcccccc@{}}
       \toprule
       \textbf{USA} & Geom 1 & Geom 5 & Geom 10 & Geom 50 & Buchin & FlexScan \\
       \midrule
       Runtime (s) & 0.29 & 0.29 & 0.29 & 0.33 & 114.0 & 1109 \\
       St.dev.     & 0.001 & 0.001 & 0.002 & 0.002 & 4.4 & -- \\
       \midrule
       \textbf{Arkansas} \\
       \midrule
       Runtime (s) & 0.21 & 0.26 & 0.27 & 0.28 & 25.3 & 8.00 \\
       Std         & 0.015 & 0.005 & 0.002 & 0.002 & 0.12 & -- \\
       \bottomrule
    \end{tabular*}
     \caption{\label{tab:runtime_usa+AR}
     Runtime in seconds for one region discovery, scan step only. Top:
    $3{,}711$ U.S. county-by-congressional-district regions; bottom: $75$ Arkansas
    counties. Geom-$k$ and Buchin report means over five runs. }
    \vspace{-3mm}
\end{table}




Table~\ref{tab:runtime_usa+AR} shows that pyScan's scan time is dominated
by the fixed-resolution grid scan, so Geom $k$ for $k \in \{1,5,10,50\}$
all complete one discovery in about $0.3$ seconds on the USA dataset
($n = 3{,}711$ regions). While we demonstrate this with pyScan, other
efficient scanning approaches, like those by Neill and Moore
\citep{NM04}, would also keep this manageable. Buchin's rectangle scanner takes $114$ seconds per discovery on the USA dataset --- about
$350\times$ slower than Geom 50 --- because it sweeps a nine-point size
grid and performs a translation search at each size. FlexScan takes $1{,}109$ seconds on the USA dataset (about $18$
minutes), although that number includes a Monte Carlo significance test the
other methods do not perform. As discussed in Section~\ref{sec:software},
FlexScan also presents interpretability challenges on large datasets,
returning multiple clusters that are not straightforward to reconcile.


To complement our large-scale analysis, we also evaluated runtimes
on the smaller Arkansas dataset, using a planted region that covers
approximately $30\%$ of the state; bottom part of Table \ref{tab:runtime_usa+AR}. 
This setup aligns with prior evaluation studies \citep{tango2005flexibly,tango2012restrict} using
FlexScan. Geom $50$ completes a single discovery in approximately
$0.28$ seconds; Buchin's rectangle scanner takes $25$ seconds (about $90\times$ slower), and FlexScan takes $8$ seconds, including its
Monte Carlo significance test. These results demonstrate that our method's runtime advantage
holds across both small and large region counts.  


Finally,  pyScan also offers more aggressive approximation methods that compute the approximately optimal spatial scan statistic.  These employ algorithms that involve subsampling data in addition to adaptive gridding, and other methods to scan them efficiently.  In particular, these methods use coresets~\citep{phillips2017coresets}, which first reduce the dataset problem to a size that depends only on the desired accuracy guarantee.  As a result, the initial size of the dataset no longer becomes a runtime constraint (other than the time to store and sample from it), only the desired degree of accuracy.  In this context, the proposed approach will not meaningfully impact the runtime.


\subsection{Convergence}
\label{sec:converge}

Our experiments demonstrate that sampling multiple points per region improves the statistical power, enabling more accurate recovery of true anomalous regions. However, we have not yet provided a theoretical explanation of why these methods work. In this section, we offer a stylized analysis to provide such insight.
As with any stylized analysis, it will be an imperfect model of real-world settings, but we hope it provides some insight nonetheless. 
There are two key components to understanding how accuracy improves with the number $k$ of sample points per region in \textsf{Geom $k$}.  First, how accurately we can approximate the  function $\phi$, which the scan statistics algorithm seeks to optimize. Second, how this approximation translates to the quality of the recovered region, measured using Jaccard distance $d_{\textsc{Jac}}$.  We address the second issue first.  The scan statistic problem is a non-convex optimization problem with the potential for many local minima (e.g., see conditional hardness results~\citep{MP18,backurs2016tight}), so unless we make some structural assumptions about the problem, even small approximation errors in $\phi$ could lead to a significantly different optimal shape, resulting in a large change in Jaccard distance.  
To address this, we analyze what we call \emph{($\eps,\gamma)$-stable shapes}. An optimal shape $S^* \in \mathcal{S}$ is ($\eps,\gamma$)-stable if for any $S' \in \mathcal{S}$ such that $d_{\textsc{Jac}}(S^*, S') > \gamma$ it holds that $\phi(S^*) - \phi(S') > \eps$. This stability assumption implies that any shape with a score within $\eps$ of the optimum must be within $\gamma$ in Jaccard distance.

While this assumption cannot be guaranteed for all real-world scenarios, in our simulated setting -- with its controlled random generation process 
on a known and fixed planted shape $S$ 
-- we can satisfy that the optimal region $S^*$ is $(\eps,\gamma)$-stable for some choice of $\eps$ and $\gamma$; at least with a large enough sample size and $pq$ difference.  Thus, for our theoretical analysis, we assume to be seeking an $(\eps,\gamma)$-stable shape; we will show that we can achieve $\eps$-approximation in the cost $\phi$, and this will imply we also obtain $\gamma$-approximation in the Jaccard distance.  

To ensure that the recovered region $(\hat S)$ is close to the true optimal region $(S^*)$, we require that the difference in their scan statistic score is at most $\eps$, that is
\[
|\phi(\hat{S}) - \phi(S^*)| \leq \eps.  
\]
In other words, we want the discovered region $\hat S$ to approximate the optimal region $S^*$, within a small margin of error in terms of the score function $\phi$.  
To apply this in our setting, we need one additional concept, since 
the regions are imbalanced in size.   Hence,  
we may need more points from the larger regions (e.g., more samples from Texas than Delaware in a US state dataset).  
To allow for analysis of our proposed algorithm, we consider an input of a set of regions $\mathcal{R}$ to be \emph{$\kappa$-balanced} if $\max_{R \in \mathcal{R}} m(R) \leq \kappa \min_{R \in \mathcal{R}} m(R)$ and $\max_{R \in \mathcal{R}} b(R) \leq \kappa \min_{R \in \mathcal{R}} b(R)$. 
We can now state the following.
 
\begin{theorem}
\label{thm:geomk}
    Consider an $(\eps,\gamma)$-stable anomalous shape $S^*$ in a $\kappa$-balanced set of $n$ regions.  If we use \textsf{Geom} $k$ to sample $k = O(\kappa / (n \eps^2))$ uniformly distributed points from each region and then run a scan statistics algorithm to obtain a shape $\hat S$, then with a constant probability we can guarantee $d_{\textsc{Jac}}(S^*,\hat S) \leq \gamma$.  
\end{theorem}
\begin{proof}
We rely on 
a probabilistic guarantee established in prior work~\citep{SSSS,MP18} that we can ensure an accurate estimate of $\phi$ under sampling.  
%
%
%
In particular, if we consider a random sample $X' \subset X$ of size $|X'| = O(1/\eps^2)$, then with constant probability, the optimal scan statistic computed on $X'$ will have a solution $\hat S = S'^*$ so that $|\phi(S'^*) - \phi(S^*) | \leq \eps$, where $S^*$ is the optimal solution of the full dataset $X$.   

To apply these results in our settings, we treat the input $X$ as a continuous distribution -- in particular, a uniform measure over each input region.  In our application, $X$ represents either the measured value $m
$ or the baseline value $b$; where the density assigned to each region $z$ is given by $m(z)$ or $b(z)$, respectively.  
There is nothing in the above results~\citep{SSSS,MP18} that prevents $X$ from being a continuous distribution, as long as we can draw random samples from it.  Since we can draw uniformly from each shape file region, this modeling assumption is justified.   
The remaining challenge is that the required $O(1/\eps^2)$ samples must be drawn from the full dataset, proportionally to the baseline $b$ (and likewise $m$). This means that if there are $n$ regions, we require $kn = O(1/\eps^2)$ total samples, $k = O(1/(n \eps^2))$ points -- indicating as we subdivide into a larger number of regions $n$, the total number of samples stays fixed, and thus the number of samples needed per region decreases.  

However, this $k = O(1/(n \eps^2))$ bound assumes the regions are comparable in size.  
However, since we assume the regions are $\kappa$-balanced, then, if we sample $k = O(\kappa / (n \eps^2))$ points per region, it will over-sample from some smaller regions, but guarantee to obtain the equivalent of the requisite (at most $\kappa / \eps^2$) number of samples even in the largest regions.  

Piecing together these components and transferring $\eps$-error on $\phi$ to $\gamma$-Jaccard error via the $(\eps,\gamma)$-stable assumption completes the proof.  
\end{proof}

Perhaps surprisingly, this analysis indicates that as the number of regions $n$ increases, then the number of samples per region $k$ (the parameter in Geom $k$), \emph{decreases}, proportional to $(\kappa/\eps^2)/n$.  
To make sense of this, if the area of study (say all of the continental USA) is fixed, but the size of each region decreases (e.g., from states to counties to zip codes), then it is already gaining more resolution.  Thus to cover the same area, there are fewer points per region $k$ needed to allow the scanning to obtain the same refinement.

\subsection{Choice of $k$}
\label{sec:k-choice}

We further empirically explore the effect of the choice of $k$ (uniform samples per region) and the recommendation of $k \in [20, 50]$ as effective and efficient.   
We sweep $k$ from $2$ to $100$ on all six datasets in Figure~\ref{fig:k-sweep}.  
So a state with $n$ regions contributes $kn$ total points to the scan input.  We fix the rate contrast at the $pq$ difference $= 0.15$ and report the Point Jaccard distance averaged over $60$ trials.  

For most datasets Point Jaccard distance drops sharply until $k \approx 20$, after which returns diminish, although the large $n=3{,}108$ USA dataset is already below the $0.2$ recovery threshold by $k=3$.

\begin{figure}[b]
    \includegraphics[width=0.85\linewidth]{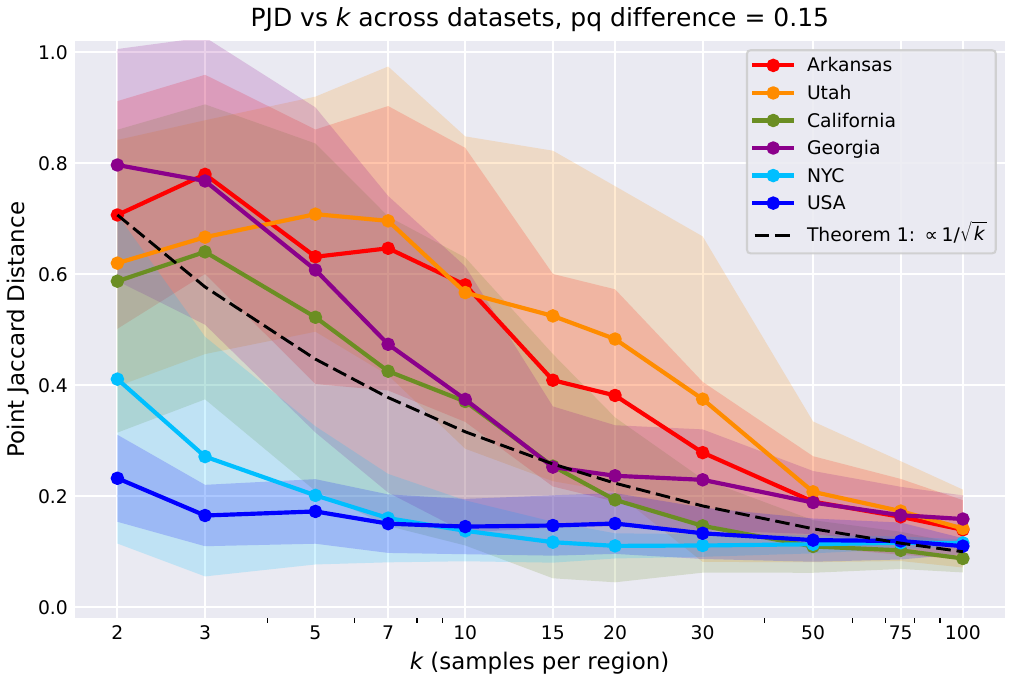}
    \Description{A line plot of Point Jaccard distance (y-axis, 0 to 1)
versus k samples per region (x-axis on a log scale, with tick labels 2,
3, 5, 7, 10, 15, 20, 30, 50, 75, 100). Six colored curves are shown,
one per dataset: Arkansas, Utah, California, Georgia, NYC, and USA,
each with shaded standard-deviation bands. All six curves generally
decrease as k grows. NYC, Georgia, and the U.S. dataset level off after
approximately k equal to 20, while Arkansas, Utah, and California
continue decreasing through k equal to 100; Utah shows the largest
improvement after k equal to 20. A black dashed reference line, labeled in the legend as proportional to 1 over square root of k, follows the same
decreasing shape and overlays the dataset curves past k roughly equal
to 10.}

    \vspace{-3mm}
     \caption{Point Jaccard distance as a function of $k$ on Geom $k$, at fixed $pq$ difference of $0.15$, across six datasets. 
 Curves show means over $60$ trials; shaded
    bands show $\pm1$ std.dev.}
    \label{fig:k-sweep}
\end{figure}
  


Note that Theorem \ref{thm:geomk}  gives that $k = (\kappa/n)(1/\eps^2)$ samples are sufficient to achieve $\eps$ error (for a fixed dataset, with $\kappa, n$ fixed).  
That means error ($\eps$) drops proportional to $1/\sqrt{k}$.  
This rate is also plotted in Figure \ref{fig:k-sweep} as dashed black, pegged to the Arkansas dataset at $k=2$.  
We observe that this general shape roughly applies to the various datasets; however the constants differ a bit as does $n$; with the large $n$ datasets (USA and NYC) requiring fewer samples, as expected.  Moreover, this does not converge to zero: by $k=100$ the curves flatten between about $0.09$ and $0.16$.  
There is other variation that is harder to pinpoint such as that captured by $\kappa$.  


\section{Real Example: California Valley Fever}
\label{app:real}

To validate the effect on real data, we run on county-level Valley
Fever (coccidioidomycosis) incidence in California, whose elevated region has a
known cause independent of the case counts. The fungal pathogen
\emph{Coccidioides} lives in arid soils, and its California endemic zone of 
San Joaquin Valley (SJV), centered on Kern County, and mapped from soil ecology~\cite{cdc-cocci,cdph-cocci}. We set $m(z)$ to confirmed cases
and $b(z)$ to population per county (CDPH, 2014--2018~\cite{chhs-idb}), take the
ground truth $S^\star$ to be the eight SJV counties, run the Kulldorff scan over
axis-aligned rectangles, and score each discovered region by its Point Jaccard
distance to $S^\star$. 


\begin{figure}[t]
  \centering
  \includegraphics[width=1.03\linewidth]{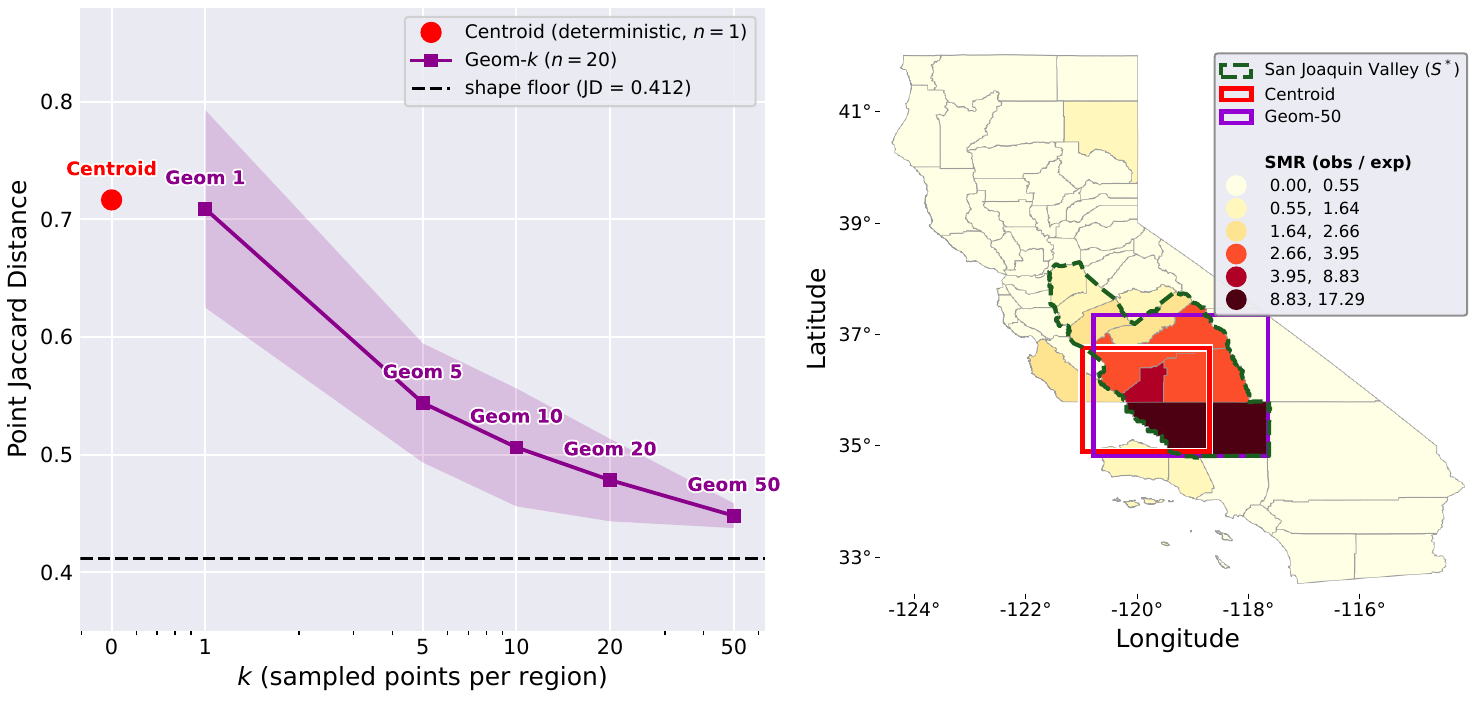}
  \Description{Left: Point Jaccard distance to the San Joaquin Valley reference region, for Centroid and for Geom-k as k increases from 1 to 50, with a dashed horizontal line marking the best-fit rectangle benchmark. Right: a choropleth map of California counties shaded by standardized morbidity ratio, with the San Joaquin Valley reference counties outlined and the rectangles returned by Centroid and Geom-50 drawn on top.}

  \vspace{-3mm}
  \caption{California Valley Fever: Point Jaccard distance to the San Joaquin
  Valley ground truth. The centroid ($k{=}0$, red) is shown
  against Geom-$k$ for $k\in\{1,5,10,20,50\}$ (purple, mean $\pm 1$ std band),
  shape floor (dashed). Right shows
  standardized morbidity ratio (SMR, observed/expected cases) by county, with
  the SJV ground truth $S^\star$ outlined (green dashed) and
  the regions found by Centroid and Geom-50.}
  \label{fig:vf-curve}
\end{figure}


The centroid attains $\mathrm{PJD}=0.717$; Geom-$k$ improves monotonically to
$0.448$ at $k=50$, a gap of $+0.268$ (Figure~\ref{fig:vf-curve}). This lands only
$0.036$ above the shape floor of $0.412$: the best overlap an
axis-aligned rectangle can achieve, since the diagonally oriented valley cannot
be wrapped tightly.  In contrast, the centroid is stranded $0.305$ above it. The
discovered region sits inside the SJV with standardized morbidity ratio
(observed/expected cases) $\approx 8$ at every $k$,
and the per-trial spread
narrows as $k$ grows.
The valley's
few very large counties (notably Kern) carry most of the case mass; collapsing
each to its centroid destroys the information needed to \emph{size} the region,
while sampling restores it.

The improvement is robust to both natural design choices. 
Tightening $S^\star$ to the five hyperendemic core counties widens it to $+0.285$ (floor $0.294$).  
Varying the window across a pre-surge period
(2011--2013), the headline period (2014--2018), and the peak-outbreak year (2017) keeps it between $+0.243$ and $+0.268$, despite a $1.8\times$ swing in annual case
load. The separation therefore reflects region geometry, not case volume.

\section{Conclusion}
\label{sec:conclusion}

We introduce a simple yet effective method for converting region-based inputs to point-based representation, enabling efficient algorithms for spatial scan statistics.  Our approach,  sampling multiple points uniformly within each region, significantly improves statistical power compared to standard techniques that rely solely on region centroids. Through planted-region experiments, we demonstrated that our method can recover anomalous regions under more subtle effect sizes while remaining computationally tractable. 

Leveraging efficient implementation such as pyScan, we showed that even with increased point counts (e.g., Geom 50), runtime remains practical across both small and large datasets. This supports the scalability of our method without sacrificing precision or speed.
  
Given its ease of implementation, superior performance, and compatibility with existing algorithms, we recommend this sampling-based approach as the default way to convert region-based input to point-based representations for spatial scan statistics.  





\clearpage 
\bibliographystyle{ACM-Reference-Format}
\bibliography{references}


\begin{thebibliography}{28}


\ifx \showCODEN    \undefined \def \showCODEN     #1{\unskip}     \fi
\ifx \showDOI      \undefined \def \showDOI       #1{#1}\fi
\ifx \showISBNx    \undefined \def \showISBNx     #1{\unskip}     \fi
\ifx \showISBNxiii \undefined \def \showISBNxiii  #1{\unskip}     \fi
\ifx \showISSN     \undefined \def \showISSN      #1{\unskip}     \fi
\ifx \showLCCN     \undefined \def \showLCCN      #1{\unskip}     \fi
\ifx \shownote     \undefined \def \shownote      #1{#1}          \fi
\ifx \showarticletitle \undefined \def \showarticletitle #1{#1}   \fi
\ifx \showURL      \undefined \def \showURL       {\relax}        \fi
\providecommand\bibfield[2]{#2}
\providecommand\bibinfo[2]{#2}
\providecommand\natexlab[1]{#1}
\providecommand\showeprint[2][]{arXiv:#2}

\bibitem[Abolhassani and Prates(2021)]%
        {abolhassani2021up}
\bibfield{author}{\bibinfo{person}{Ali Abolhassani} {and}
  \bibinfo{person}{Marcos~O Prates}.} \bibinfo{year}{2021}\natexlab{}.
\newblock \showarticletitle{An up-to-date review of scan statistics}.
\newblock \bibinfo{journal}{\emph{Statistic Surveys}}  \bibinfo{volume}{15}
  (\bibinfo{year}{2021}), \bibinfo{pages}{111--153}.
\newblock


\bibitem[Agarwal et~al\mbox{.}(2006a)]%
        {agarwal2006spatial}
\bibfield{author}{\bibinfo{person}{Deepak Agarwal}, \bibinfo{person}{Andrew
  McGregor}, \bibinfo{person}{Jeff~M Phillips}, \bibinfo{person}{Suresh
  Venkatasubramanian}, {and} \bibinfo{person}{Zhengyuan Zhu}.}
  \bibinfo{year}{2006}\natexlab{a}.
\newblock \showarticletitle{Spatial scan statistics: approximations and
  performance study}. In \bibinfo{booktitle}{\emph{Proceedings of the 12th ACM
  SIGKDD international conference on Knowledge discovery and data mining}}.
  \bibinfo{pages}{24--33}.
\newblock


\bibitem[Agarwal et~al\mbox{.}(2006b)]%
        {APV06}
\bibfield{author}{\bibinfo{person}{Deepak Agarwal}, \bibinfo{person}{Jeff~M.
  Phillips}, {and} \bibinfo{person}{Suresh Venkatasubramanian}.}
  \bibinfo{year}{2006}\natexlab{b}.
\newblock \showarticletitle{The Hunting of the Bump: On Maximizing Statistical
  Discrepancy}.
\newblock \bibinfo{journal}{\emph{SODA}} (\bibinfo{year}{2006}),
  \bibinfo{pages}{1137--1146}.
\newblock


\bibitem[Backurs et~al\mbox{.}(2016)]%
        {backurs2016tight}
\bibfield{author}{\bibinfo{person}{Arturs Backurs}, \bibinfo{person}{Nishanth
  Dikkala}, {and} \bibinfo{person}{Christos Tzamos}.}
  \bibinfo{year}{2016}\natexlab{}.
\newblock \showarticletitle{Tight Hardness Results for Maximum Weight
  Rectangles}. In \bibinfo{booktitle}{\emph{43rd International Colloquium on
  Automata, Languages, and Programming (ICALP 2016)}},
  Vol.~\bibinfo{volume}{55}. \bibinfo{publisher}{Schloss Dagstuhl},
  \bibinfo{address}{Dagstuhl, Germany}, \bibinfo{pages}{81}.
\newblock


\bibitem[Buchin et~al\mbox{.}(2012)]%
        {buchin2012}
\bibfield{author}{\bibinfo{person}{Kevin Buchin}, \bibinfo{person}{Maike
  Buchin}, \bibinfo{person}{Marc van Kreveld}, \bibinfo{person}{Maarten
  L{\"o}ffler}, \bibinfo{person}{Jun Luo}, {and} \bibinfo{person}{Rodrigo~I.
  Silveira}.} \bibinfo{year}{2012}\natexlab{}.
\newblock \showarticletitle{Processing aggregated data: the location of
  clusters in health data}.
\newblock \bibinfo{journal}{\emph{GeoInformatica}} \bibinfo{volume}{16},
  \bibinfo{number}{3} (\bibinfo{year}{2012}), \bibinfo{pages}{497--521}.
\newblock
\urldef\tempurl%
\url{https://doi.org/10.1007/s10707-011-0143-6}
\showDOI{\tempurl}


\bibitem[{California Department of Public Health}(2024a)]%
        {cdph-cocci}
\bibfield{author}{\bibinfo{person}{{California Department of Public Health}}.}
  \bibinfo{year}{2024}\natexlab{a}.
\newblock \bibinfo{title}{Coccidioidomycosis ({Valley Fever})}.
\newblock
  \bibinfo{howpublished}{\url{https://www.cdph.ca.gov/Programs/CID/DCDC/Pages/Coccidioidomycosis.aspx}}.
\newblock
\newblock
\shownote{Accessed 2026-06-05}.


\bibitem[{California Department of Public Health}(2024b)]%
        {chhs-idb}
\bibfield{author}{\bibinfo{person}{{California Department of Public Health}}.}
  \bibinfo{year}{2024}\natexlab{b}.
\newblock \bibinfo{title}{Infectious Diseases by Disease, County, Year, and
  Sex}.
\newblock \bibinfo{howpublished}{California Health and Human Services Open Data
  Portal.
  \url{https://data.chhs.ca.gov/dataset/03e61434-7db8-4a53-a3e2-1d4d36d6848d}}.
\newblock
\newblock
\shownote{Accessed 2026-06-05}.


\bibitem[{Centers for Disease Control and Prevention}(2024)]%
        {cdc-cocci}
\bibfield{author}{\bibinfo{person}{{Centers for Disease Control and
  Prevention}}.} \bibinfo{year}{2024}\natexlab{}.
\newblock \bibinfo{title}{Valley Fever ({Coccidioidomycosis}): Areas Where It
  Lives}.
\newblock \bibinfo{howpublished}{\url{https://www.cdc.gov/valley-fever/}}.
\newblock
\newblock
\shownote{Accessed 2026-06-05}.


\bibitem[Desjardins et~al\mbox{.}(2020)]%
        {desjardins2020rapid}
\bibfield{author}{\bibinfo{person}{Michael~R Desjardins},
  \bibinfo{person}{Alexander Hohl}, {and} \bibinfo{person}{Eric~M Delmelle}.}
  \bibinfo{year}{2020}\natexlab{}.
\newblock \showarticletitle{Rapid surveillance of COVID-19 in the United States
  using a prospective space-time scan statistic: Detecting and evaluating
  emerging clusters}.
\newblock \bibinfo{journal}{\emph{Applied geography}}  \bibinfo{volume}{118}
  (\bibinfo{year}{2020}), \bibinfo{pages}{102202}.
\newblock


\bibitem[Glaz and Koutras(2024)]%
        {glaz2024handbook}
\bibfield{author}{\bibinfo{person}{Joseph Glaz} {and} \bibinfo{person}{Markos~V
  Koutras}.} \bibinfo{year}{2024}\natexlab{}.
\newblock \bibinfo{booktitle}{\emph{Handbook of scan statistics}}.
\newblock \bibinfo{publisher}{Springer}.
\newblock


\bibitem[Grubesic et~al\mbox{.}(2014)]%
        {grubesic}
\bibfield{author}{\bibinfo{person}{Tony~H Grubesic}, \bibinfo{person}{Ran Wei},
  {and} \bibinfo{person}{Alan~T Murray}.} \bibinfo{year}{2014}\natexlab{}.
\newblock \showarticletitle{Spatial clustering overview and comparison:
  Accuracy, sensitivity, and computational expense}.
\newblock \bibinfo{journal}{\emph{Annals of the Association of American
  Geographers}} \bibinfo{volume}{104}, \bibinfo{number}{6}
  (\bibinfo{year}{2014}), \bibinfo{pages}{1134--1156}.
\newblock


\bibitem[Han et~al\mbox{.}(2019)]%
        {han2019kernel}
\bibfield{author}{\bibinfo{person}{Mingxuan Han}, \bibinfo{person}{Michael
  Matheny}, {and} \bibinfo{person}{Jeff~M Phillips}.}
  \bibinfo{year}{2019}\natexlab{}.
\newblock \showarticletitle{The kernel spatial scan statistic}. In
  \bibinfo{booktitle}{\emph{Proceedings of the 27th ACM SIGSPATIAL
  International Conference on Advances in Geographic Information Systems}}.
  \bibinfo{publisher}{ACM}, \bibinfo{address}{New York, NY, USA},
  \bibinfo{pages}{349--358}.
\newblock


\bibitem[Kulldorff(1997)]%
        {kulldorff1997spatial}
\bibfield{author}{\bibinfo{person}{Martin Kulldorff}.}
  \bibinfo{year}{1997}\natexlab{}.
\newblock \showarticletitle{A spatial scan statistic}.
\newblock \bibinfo{journal}{\emph{Communications in Statistics-Theory and
  methods}} \bibinfo{volume}{26}, \bibinfo{number}{6} (\bibinfo{year}{1997}),
  \bibinfo{pages}{1481--1496}.
\newblock


\bibitem[Kulldorff(2022)]%
        {Kul7.0}
\bibfield{author}{\bibinfo{person}{Martin Kulldorff}.}
  \bibinfo{year}{2022}\natexlab{}.
\newblock \bibinfo{booktitle}{\emph{SatScan User Guide}
  (\bibinfo{edition}{10.1} ed.)}.
\newblock http://www.satscan.org/.
\newblock


\bibitem[Kulldorff et~al\mbox{.}(2006)]%
        {Kulldorff2006}
\bibfield{author}{\bibinfo{person}{Martin Kulldorff}, \bibinfo{person}{Lan
  Huang}, \bibinfo{person}{Linda Pickle}, {and} \bibinfo{person}{Luiz
  Duczmal}.} \bibinfo{year}{2006}\natexlab{}.
\newblock \showarticletitle{An elliptic spatial scan statistic.}
\newblock \bibinfo{journal}{\emph{Statistics in medicine}}
  \bibinfo{volume}{25}, \bibinfo{number}{22} (\bibinfo{year}{2006}),
  \bibinfo{pages}{3929--43}.
\newblock


\bibitem[Matheny(2024)]%
        {pyScan}
\bibfield{author}{\bibinfo{person}{Michael Matheny}.}
  \bibinfo{year}{2024}\natexlab{}.
\newblock \bibinfo{booktitle}{\emph{{pyScan}}}.
\newblock https://github.com/michaelmathen/pyscan.
\newblock


\bibitem[Matheny and Phillips(2018)]%
        {MP18}
\bibfield{author}{\bibinfo{person}{Michael Matheny} {and}
  \bibinfo{person}{Jeff~M. Phillips}.} \bibinfo{year}{2018}\natexlab{}.
\newblock \showarticletitle{Computing Approximate Statistical Discrepancy}.
\newblock \bibinfo{journal}{\emph{ISAAC}}  \bibinfo{volume}{123}
  (\bibinfo{year}{2018}), \bibinfo{pages}{7:1--7:14}.
\newblock


\bibitem[Matheny et~al\mbox{.}(2016)]%
        {SSSS}
\bibfield{author}{\bibinfo{person}{Michael Matheny},
  \bibinfo{person}{Raghvendra Singh}, \bibinfo{person}{Liang Zhang},
  \bibinfo{person}{Kaiqiang Wang}, {and} \bibinfo{person}{Jeff~M. Phillips}.}
  \bibinfo{year}{2016}\natexlab{}.
\newblock \showarticletitle{Scalable Spatial Scan Statistics Through Sampling}.
  In \bibinfo{booktitle}{\emph{SIGSPATIAL}}. \bibinfo{publisher}{ACM},
  \bibinfo{address}{New York, NY, USA}, \bibinfo{pages}{1--10}.
\newblock


\bibitem[Neill and Moore(2004)]%
        {NM04}
\bibfield{author}{\bibinfo{person}{Daniel~B. Neill} {and}
  \bibinfo{person}{Andrew~W. Moore}.} \bibinfo{year}{2004}\natexlab{}.
\newblock \showarticletitle{Rapid Detection of Significant Spatial Clusters}.
  In \bibinfo{booktitle}{\emph{KDD}}. \bibinfo{publisher}{ACM},
  \bibinfo{address}{New York, NY, USA}, \bibinfo{pages}{256--265}.
\newblock


\bibitem[Neill et~al\mbox{.}(2006)]%
        {NMC06}
\bibfield{author}{\bibinfo{person}{Daniel~B. Neill}, \bibinfo{person}{Andrew~W.
  Moore}, {and} \bibinfo{person}{Gregory~F. Cooper}.}
  \bibinfo{year}{2006}\natexlab{}.
\newblock \showarticletitle{A {B}ayesian Spatial Scan Statistic}. In
  \bibinfo{booktitle}{\emph{NIPS}}. \bibinfo{publisher}{MIT Press},
  \bibinfo{address}{Cambridge, MA, USA}, \bibinfo{pages}{1003--1010}.
\newblock


\bibitem[Nobles et~al\mbox{.}(2022)]%
        {nobles2022presyndromic}
\bibfield{author}{\bibinfo{person}{Mallory Nobles}, \bibinfo{person}{Ramona
  Lall}, \bibinfo{person}{Robert~W Mathes}, {and} \bibinfo{person}{Daniel~B
  Neill}.} \bibinfo{year}{2022}\natexlab{}.
\newblock \showarticletitle{Presyndromic surveillance for improved detection of
  emerging public health threats}.
\newblock \bibinfo{journal}{\emph{Science Advances}} \bibinfo{volume}{8},
  \bibinfo{number}{44} (\bibinfo{year}{2022}), \bibinfo{pages}{eabm4920}.
\newblock


\bibitem[Patil and Taillie(2004)]%
        {patil2004upper}
\bibfield{author}{\bibinfo{person}{Ganapati~P Patil} {and}
  \bibinfo{person}{Charles Taillie}.} \bibinfo{year}{2004}\natexlab{}.
\newblock \showarticletitle{Upper level set scan statistic for detecting
  arbitrarily shaped hotspots}.
\newblock \bibinfo{journal}{\emph{Environmental and Ecological statistics}}
  \bibinfo{volume}{11} (\bibinfo{year}{2004}), \bibinfo{pages}{183--197}.
\newblock


\bibitem[Phillips(2017)]%
        {phillips2017coresets}
\bibfield{author}{\bibinfo{person}{Jeff~M Phillips}.}
  \bibinfo{year}{2017}\natexlab{}.
\newblock \showarticletitle{Coresets and sketches}.
\newblock In \bibinfo{booktitle}{\emph{Handbook of discrete and computational
  geometry}}. \bibinfo{publisher}{Chapman and Hall/CRC},
  \bibinfo{pages}{1269--1288}.
\newblock


\bibitem[Shiode(2011)]%
        {shiode2011street}
\bibfield{author}{\bibinfo{person}{Shino Shiode}.}
  \bibinfo{year}{2011}\natexlab{}.
\newblock \showarticletitle{Street-level spatial scan statistic and STAC for
  analysing street crime concentrations}.
\newblock \bibinfo{journal}{\emph{Transactions in GIS}} \bibinfo{volume}{15},
  \bibinfo{number}{3} (\bibinfo{year}{2011}), \bibinfo{pages}{365--383}.
\newblock


\bibitem[Takahashi and Tango(2023)]%
        {FleXScan}
\bibfield{author}{\bibinfo{person}{Tetsuji~Yokoyama Takahashi} {and}
  \bibinfo{person}{Toshiro Tango}.} \bibinfo{year}{2023}\natexlab{}.
\newblock \bibinfo{booktitle}{\emph{{FleXScan}: Software for the Flexible Scan
  Statistics}}.
\newblock https://sites.google.com/site/flexscansoftware/.
\newblock


\bibitem[Tango and Takahashi(2005)]%
        {tango2005flexibly}
\bibfield{author}{\bibinfo{person}{Toshiro Tango} {and}
  \bibinfo{person}{Kunihiko Takahashi}.} \bibinfo{year}{2005}\natexlab{}.
\newblock \showarticletitle{A flexibly shaped spatial scan statistic for
  detecting clusters}.
\newblock \bibinfo{journal}{\emph{International journal of health geographics}}
   \bibinfo{volume}{4} (\bibinfo{year}{2005}), \bibinfo{pages}{1--15}.
\newblock


\bibitem[Tango and Takahashi(2012)]%
        {tango2012restrict}
\bibfield{author}{\bibinfo{person}{Toshiro Tango} {and}
  \bibinfo{person}{Kunihiko Takahashi}.} \bibinfo{year}{2012}\natexlab{}.
\newblock \showarticletitle{A Flexible Spatial Scan Statistic with a Restricted
  Likelihood Ratio for Detecting Clusters}.
\newblock \bibinfo{journal}{\emph{Statistics in Medicine}}
  \bibinfo{volume}{31}, \bibinfo{number}{30} (\bibinfo{year}{2012}),
  \bibinfo{pages}{4207--4218}.
\newblock
\urldef\tempurl%
\url{https://doi.org/10.1002/sim.5478}
\showDOI{\tempurl}


\bibitem[Xie et~al\mbox{.}(2022)]%
        {xie2022statistically}
\bibfield{author}{\bibinfo{person}{Yiqun Xie}, \bibinfo{person}{Shashi
  Shekhar}, {and} \bibinfo{person}{Yan Li}.} \bibinfo{year}{2022}\natexlab{}.
\newblock \showarticletitle{Statistically-robust clustering techniques for
  mapping spatial hotspots: A survey}.
\newblock \bibinfo{journal}{\emph{ACM Computing Surveys (CSUR)}}
  \bibinfo{volume}{55}, \bibinfo{number}{2} (\bibinfo{year}{2022}),
  \bibinfo{pages}{1--38}.
\newblock


\end{thebibliography}

\clearpage



\begin{appendices}

\section{Implementation and pyScan}\label{secA1}
We include these listings as a reproducibility aid: together they reproduce, in about fifty lines of Python, the full \textsf{Geom 50} pipeline used throughout Section~\ref{sec:exp} on any input shapefile, so a reader can replicate our results without reconstructing the \texttt{pyScan} calls from the body of the paper.

We illustrate how to run \texttt{pyScan} on U.S. counties using the Geom 50 setting. First, we load a shapefile of U.S. counties.

\begin{lstlisting}[caption={Python Code to Load Shapefile}, label={lst:shapefile}, language=Python]
import pyscan
import geopandas as gpd
import numpy as np
import random
from shapely.geometry import Point, Polygon
import matplotlib.pyplot as plt

# Load shapefile of US counties
us_df = gpd.read_file("us_county.shp").to_crs("EPSG:4326")

# Check the first few entries
print(us_df.head())
\end{lstlisting}

Next, we define a target rectangle over the state of interest:

\begin{lstlisting}[caption={Defining a Target Rectangle Using Polygon}, label={lst:target_rectangle}, language=Python]
# Define target rectangle
target_rectangle = Polygon([(-100, 33), (-100, 40), (-90, 40), (-90, 33)])
\end{lstlisting}

We generate 50 random points inside each region and assign all of them to the \texttt{baseline} set. Then, we construct the \texttt{measured} set probabilistically: points inside the target rectangle are included with probability $0.4$, and those outside with probability $0.2$.

\begin{lstlisting}[caption={Python Code to Generate 50 Random Points Inside Each Region}, label={lst:geom50-sample}, language=Python]
# Sample 50 points from each region
n_geom = 50
sampled_points = []

for i, row in us_df.iterrows():
    polygon = row['geometry']
    min_x, min_y, max_x, max_y = polygon.bounds
    region_points = []

    while len(region_points) < n_geom:
        # Generate a random point within the bounding box
        random_point = Point([random.uniform(min_x, max_x), random.uniform(min_y, max_y)])
        if random_point.within(polygon):
            region_points.append([random_point.x, random_point.y, i])
    
    sampled_points.append(region_points)

sampled_points = np.vstack(sampled_points)
\end{lstlisting}

\begin{lstlisting}[caption={Assigning Points to Baseline and Measured Sets Using Probabilistic Inclusion}, label={lst:geom50-bm}, language=Python]
# Create baseline and measured sets to test the algorithm using p=0.4 (inside target), q=0.2 (outside)
baseline = []
measured = []

for point in sampled_points:
    prob = random.random()
    # WPoint(weight, x, y, value): defines a weighted point with attribute for pyScan
    baseline.append(pyscan.WPoint(1.0, point[0], point[1], 1.0))
    pt = Point(point[0], point[1])

    if target_rectangle.contains(pt):
        if prob <= 0.4:
            measured.append(pyscan.WPoint(1.0, point[0], point[1], 1.0))
    else:
        if prob <= 0.2:
            measured.append(pyscan.WPoint(1.0, point[0], point[1], 1.0))

\end{lstlisting}

We use the \texttt{Kulldorff} scan statistic (Poisson model) to detect anomalous regions. Then we pass the \texttt{measure} and \texttt{baseline} lists to \texttt{pyscan.Grid()}. We then apply \texttt{pyscan.max\_subgrid()} to find the most anomalous rectangular subregion, stored in the variable \texttt{subgrid}. This rectangle represents the discovered rectangle.

\begin{lstlisting}[caption={Python Code to Run pyScan with Geom 50 Sampling}, label={lst:geom50-rn}, language=Python]
# Run pyScan to find most anomalous rectangle
rect_f = pyscan.KULLDORF
grid = pyscan.Grid(100, measured, baseline)
subgrid = pyscan.max_subgrid(grid, rect_f)
rect = grid.toRectangle(subgrid)
\end{lstlisting}

To visualize the result, we plot the detected rectangle over the sampled points:

\begin{lstlisting}[caption={Python Code to Plot the Detected Rectangle Over Sampled Points}, label={lst:plot_rectangle}, language=Python]
# Plot discovered rectangle
plt.scatter(sampled_points[:, 0], sampled_points[:, 1], s=2)
plt.gca().add_patch(
    plt.Rectangle((rect.lowX(), rect.lowY()),
                  rect.upX() - rect.lowX(),
                  rect.upY() - rect.lowY(),
                  edgecolor='red', facecolor='none')
)
plt.axis('off')
plt.show()
\end{lstlisting}

Full examples and further documentation for \texttt{pyScan} are available at: \\
\url{https://mmath.dev/pyscan}

\section{Disk Cluster Recovery vs.\ Buchin}
\label{app:disk-buchin}

We repeat the head-to-head comparison with the area-based method of Buchin et al.~\citep{buchin2012} using \emph{circular} planted targets.  We plant disks of radius approximately $67$ km and $44$ km, centered in the state, and follow the same protocol as the rectangular comparison in Section~\ref{sec:flexscan-buchin}: $q = 0.2$, Point Jaccard distance averaged over $20$ trials with $\pm 1$ standard deviation bands, Buchin given its standard nine-point size grid, and Centroid included as a lower reference.

\begin{figure}[htbp]
  \centering
\includegraphics[width=\linewidth]{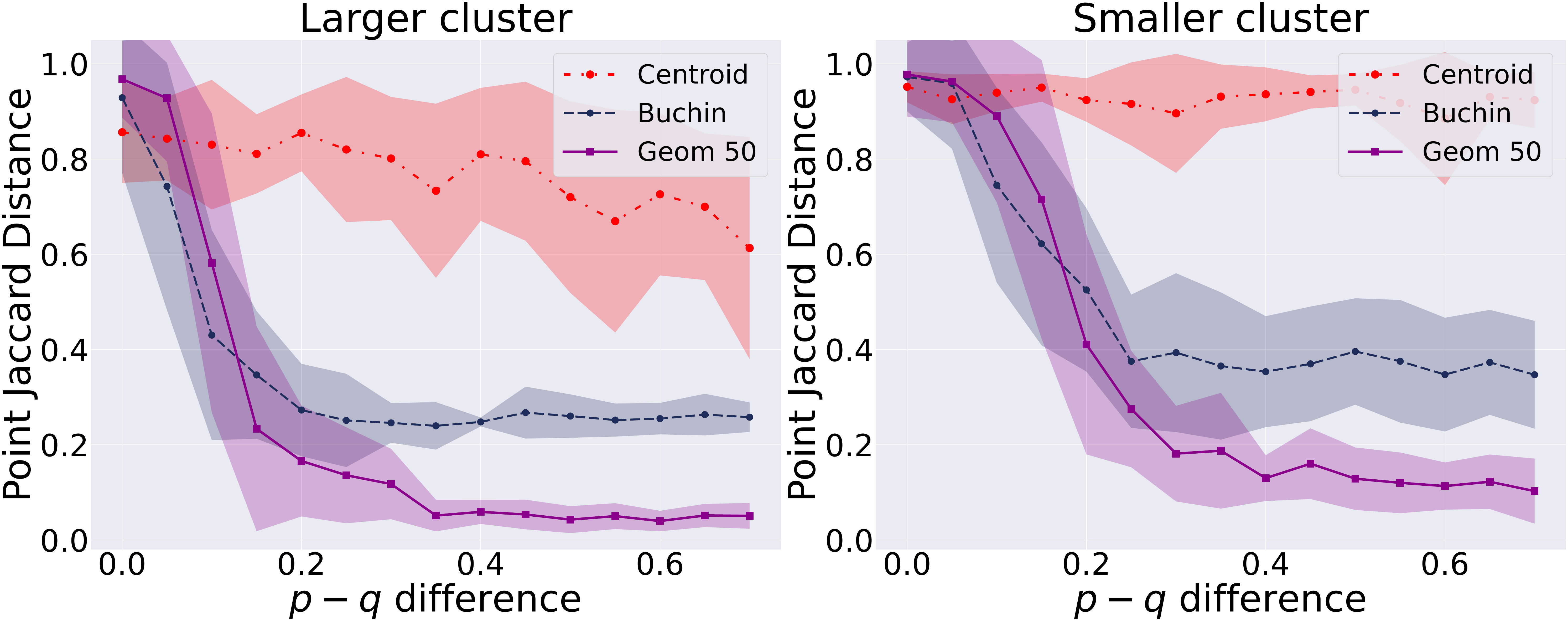}
\Description{Two side-by-side line plots of Jaccard distance (y-axis,
0 to 1) versus signal strength p minus q (x-axis, 0 to 0.7), each with
three curves: Centroid (gray), Buchin disk 9-bin (green), and Geom-50
disk (blue), with shaded standard-deviation bands. The left panel,
labeled larger cluster radius 0.6 degrees, shows the Geom-50 curve
dropping fastest and reaching near 0 by the right end, the Buchin curve
plateauing around 0.2, and the Centroid curve staying high near 0.8
across the whole range. The right panel, labeled smaller cluster radius
0.4 degrees, shows a similar pattern with a wider gap: Geom-50 reaches
near 0, Buchin plateaus around 0.3, and Centroid stays near 0.8.}
  \caption{Circular cluster recovery on Arkansas counties:
  Point Jaccard distance vs.\ $pq$ difference for planted disks of
  radius $\sim\!67$ km (left) and $\sim\!44$ km (right).
  Bands show $\pm 1$ standard deviation over $20$ trials.}
  \label{fig:DiskCurves}
\end{figure}

Figure~\ref{fig:DiskCurves} shows the recovery curves; the pattern mirrors the rectangular case.  On the larger disk the Buchin method is close to the correct radius ($1.17\times$ planted); on the smaller disk it oversizes more ($1.25\times$), where its discrete grid adapts less well. \textsf{Geom 50} recovers the correct radius at both sizes ($1.02\times$ on average), so the separation concentrates on the smaller, harder target. Buchin's residual oversizing occurs under its own nine-point size-selection protocol with the correct scale among the candidates, so it stems from the area-weighted score itself rather than any imposed restriction on the size search.  Centroid degenerates under an unbounded disk search on one point per county, confirming its unsuitability as anything but a lower reference.

Figure~\ref{fig:DiskMapAllMethods} shows the discovered windows geographically.  \textsf{Geom 50} matches (or improves over) the Buchin method's accuracy on easy (large) targets, and visually improves over it on hard (small) targets for both window shapes, and, unlike the area-based method, requires no user-supplied window size.  
Specifically on the large disk, Centroid degenerates to on average $2.17\times$ the planted radius (offset $\sim\!83.8$ km), Buchin
  selects nearly the correct radius at about $1.16\times$ (offset
  $\sim\!7.2$ km), and Geom 50 recovers the disk at $1.01\times$ (offset $\sim\!1.8$ km).  On the small disk, Centroid degenerates to $6.57\times$ planted (offset $\sim\!205.2$ km), Buchin oversizes to $1.25\times$ (offset $\sim\!10.8$ km), and Geom 50 stays at $1.03\times$ (offset $\sim\!3.6$ km).

We note that pyScan's exact disk scan enumerates candidate disks through triples of points and becomes more expensive at a national scale.  
Note that all windows are defined in EPSG:4326 lon/lat degree-Cartesian coordinates; the resulting geographic ellipticity ($1^{\circ}$ lon $\approx 91$ km vs.\ $1^{\circ}$ lat $\approx 111$ km at this latitude) applies uniformly to all methods.  The Buchin method's circular window is rendered as an inscribed regular polygon; our overlap measure treats it as a true circle, a $\sim\!2\%$ area difference in its favor. Centroid runs on an independent random stream; comparisons rely on mean-over-trials equivalence, not paired trials.

\begin{figure}[htbp]
  \includegraphics[width=0.75\linewidth]{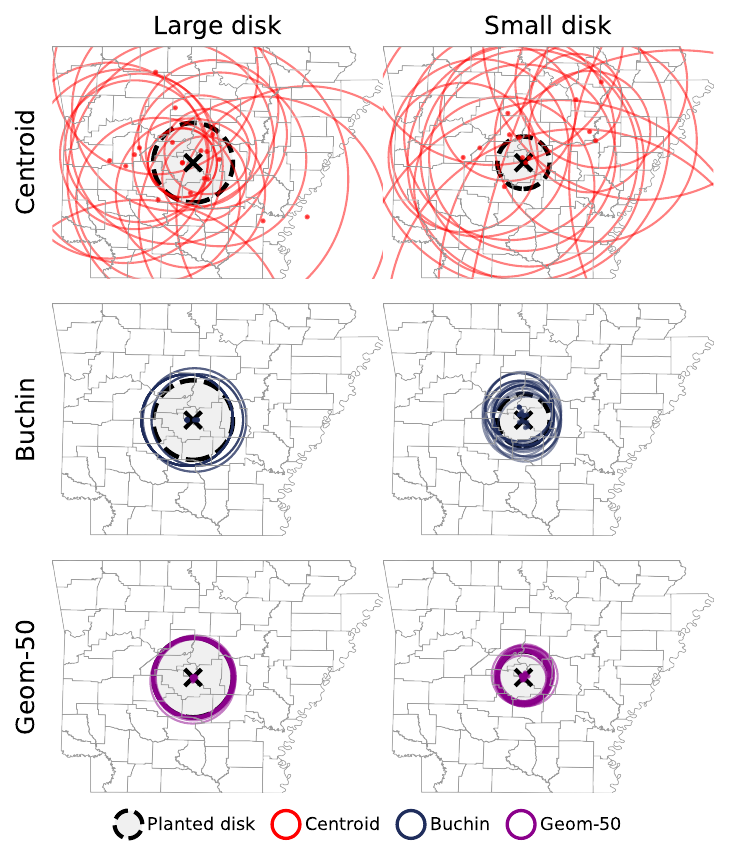}
  \Description{A 3-by-2 grid of maps of Arkansas with county boundaries
drawn in light gray. Each map has a black dashed circle marking the
planted disk target with a small black x at its center. The rows are
labeled Centroid, Buchin, and Geom-50; the columns are labeled Large
disk and Small disk. Each map shows about twenty overlapping circles
in that row's color: red for Centroid, navy for Buchin, magenta for
Geom-50. In the top row (Centroid) the circles spread wildly across
the state at varying sizes on the large disk and are dramatically
larger than the target and scattered on the small disk. In the middle
row (Buchin) the circles nest closely just outside the planted disk on
the large disk and form a thick ring noticeably larger than the
planted disk on the small disk. In the bottom row (Geom-50) the
circles cluster tightly on the planted disk in both columns.}
  \caption{Discovered disk windows on Arkansas for both planted
  disk targets (black dashed) at $p-q = 0.5$, $20$ trials each,
  planted center $(-92.5,\,34.75)$.  
  \textbf{Columns:} large disk
  (radius $\sim\!67$ km) and small disk (radius $\sim\!44$ km).
  }
  \label{fig:DiskMapAllMethods}
\end{figure}

\section{Further Discussion}
\label{sec:discuss}


This paper uses a model where the data in each region $z$ is uniformly mapped to the area in that region through a sample.  However, the human population is certainly not distributed uniformly in spatial regions, and one could think of other ways to distribute the measured and baseline values.  This could make models more realistic in baseline values but perhaps less realistic in measured values.  Our view in this paper is that when data are aggregated to regions, this fine-grained spatial information is lost, and uniform is a reasonable approach, but future work may address this more carefully.  Moreover, Section~\ref{ablation} evaluated a population-weighted alternative 
(``Weighted Sampling''), which did not outperform uniform sampling. This sensitivity check does not isolate population weighting at a fixed point budget and does not model nonuniform population density within counties; both remain directions for future work.


There are also other ways we could assess statistical power.  This could involve how we evaluate similar rectangles, how we distribute our $50 n$ samples, or experiments other than changing the $pq$ difference or the target region size.  We informally explored other ideas beyond the ones reported in this paper, and they either gave results similar to the ones we showed or were far less reliable in measurement.  The paper presents the ones we feel most clearly demonstrate the merits of the main conceptual frameworks.  


\paragraph{Generative AI Statement:}
ChatGPT and Claude were utilized to generate some text suggestions, help with formatting, refine figures, update and compile some code, and facilitate some evaluations.  All results were reviewed, verified, and approved by the authors.

\end{appendices}

\end{document}